\documentclass{article}

\usepackage{iclr2022}

\usepackage{times}
\usepackage{latexsym}
\usepackage{microtype}
\usepackage{amsfonts}
\usepackage{amsthm}
\usepackage{amsmath}
\usepackage{graphicx}
\usepackage{subfigure}
\usepackage{booktabs}
\usepackage{multirow}
\usepackage{tabularx}
\usepackage{url}
\usepackage{hyperref}
\usepackage{hyperref}
\usepackage{url}

\DeclareMathOperator*{\argmin}{arg\,min}
\DeclareMathOperator{\Tr}{Tr}
\DeclareMathOperator{\Var}{Var}
\DeclareMathOperator{\Law}{Law}
\DeclareMathOperator{\supp}{supp}
\DeclareMathOperator{\ide}{I}
\newtheorem{theorem}{Theorem}
\newtheorem{lemma}{Lemma}

\title{Diffusion-Based Voice Conversion with Fast Maximum Likelihood Sampling Scheme}

\author{Vadim Popov, Ivan Vovk, Vladimir Gogoryan, \\
Huawei Noah's Ark Lab, Moscow, Russia\\
\texttt{\{vadim.popov,vovk.ivan,gogoryan.vladimir\}@huawei.com}
\AND Tasnima Sadekova, Mikhail Kudinov \& Jiansheng Wei \\
Huawei Noah's Ark Lab, Moscow, Russia\\
\texttt{\{sadekova.tasnima,kudinov.mikhail,weijiansheng\}@huawei.com}
}

\iclrfinalcopy

\begin{document}

\maketitle

\begin{abstract}
Voice conversion is a common speech synthesis task which can be solved in different ways depending on a particular real-world scenario. The most challenging one often referred to as one-shot many-to-many voice conversion consists in copying target voice from only one reference utterance in the most general case when both source and target speakers do not belong to the training dataset. We present a scalable high-quality solution based on diffusion probabilistic modeling and demonstrate its superior quality compared to state-of-the-art one-shot voice conversion approaches. Moreover, focusing on real-time applications, we investigate general principles which can make diffusion models faster while keeping synthesis quality at a high level. As a result, we develop a novel Stochastic Differential Equations solver suitable for various diffusion model types and generative tasks as shown through empirical studies and justify it by theoretical analysis. The code is publicly available at \url{https://github.com/huawei-noah/Speech-Backbones/tree/main/DiffVC}.
\end{abstract}

\section{Introduction}
\label{sec:intro}
Voice conversion (VC) is the task of copying the target speaker's voice while preserving the linguistic content of the utterance pronounced by the source speaker. Practical VC applications often require a model which is able to operate in one-shot mode (i.e. when only one reference utterance is provided to copy the target speaker's voice) for any source and target speakers. Such models are usually referred to as one-shot many-to-many models (or sometimes zero-shot many-to-many models, or just any-to-any VC models). It is challenging to build such a model since it should be able to adapt to a new unseen voice having only one spoken utterance pronounced with it, so it was not until recently that successful one-shot VC solutions started to appear.

Conventional one-shot VC models are designed as autoencoders whose latent space ideally contains only the linguistic content of the encoded utterance while target voice identity information (usually taking shape of speaker embedding) is fed to the decoder as conditioning. Whereas in the pioneering AutoVC model \citep{AutoVC} only speaker embedding from the pre-trained speaker verification network was used as conditioning, several other models improved on AutoVC enriching conditioning with phonetic features such as pitch and loudness \citep{AutoVC-F0, AutoVC-PPG}, or training voice conversion and speaker embedding networks jointly \citep{AdaIN}. Also, several papers \citep{FragmentVC, AttentionEmbeds, PPG-VC} made use of attention mechanism to better fuse specific features of the reference utterance into the source utterance thus improving the decoder performance. Apart from providing the decoder with sufficiently rich information, one of the main problems autoencoder VC models face is to disentangle source speaker identity from speech content in the encoder. Some models \citep{AutoVC, AutoVC-F0, AutoVC-PPG} solve this problem by introducing an information bottleneck. Among other popular solutions of the disentanglement problem one can mention applying vector quantization technique to the content information \citep{VQVCp, VQMIVC}, utilizing features of Variational AutoEncoders \citep{DisentVAE, VAE-PPG, AdaIN}, introducing instance normalization layers \citep{AdaIN, AGAIN}, and using Phonetic Posteriorgrams (PPGs) \citep{AutoVC-PPG, PPG-VC}.

The model we propose in this paper solves the disentanglement problem by employing the encoder predicting ``average voice'': it is trained to transform mel features corresponding to each phoneme into mel features corresponding to this phoneme averaged across a large multi-speaker dataset. As for decoder, in our VC model, it is designed as a part of a Diffusion Probabilistic Model (DPM) since this class of generative models has shown very good results in speech-related tasks like raw waveform generation \citep{WaveGrad, DiffWave} and mel feature generation \citep{Grad-TTS, DiffTTS}. However, this decoder choice poses a problem of slow inference because DPM forward pass scheme is iterative and to obtain high-quality results it is typically necessary to run it for hundreds of iterations \citep{DDPM, IDDPM}. Addressing this issue, we develop a novel inference scheme that significantly reduces the number of iterations sufficient to produce samples of decent quality and does not require model re-training. Although several attempts have been recently made to reduce the number of DPM inference steps \citep{DDIM, NoiseEstimation, TrainableSampling, FastSampling, WaveGrad}, most of them apply to some particular types of DPMs. In contrast, our approach generalizes to all popular kinds of DPMs and has a strong connection with likelihood maximization.

This paper has the following structure: in Section \ref{sec:vc_model} we present a one-shot many-to-many VC model and describe DPM it relies on; Section \ref{sec:ml_inference} introduces a novel DPM sampling scheme and establishes its connection with likelihood maximization; the experiments regarding voice conversion task as well as those demonstrating the benefits of the proposed sampling scheme are described in Section \ref{sec:exp}; we conclude in Section \ref{sec:outro}.

\section{Voice conversion diffusion model}
\label{sec:vc_model}
As with many other VC models, the one we propose belongs to the family of autoencoders. In fact, any conditional DPM with data-dependent prior (i.e. terminal distribution of forward diffusion) can be seen as such: forward diffusion gradually adding Gaussian noise to data can be regarded as encoder while reverse diffusion trying to remove this noise acts as a decoder. DPMs are trained to minimize the distance (expressed in different terms for different model types) between the trajectories of forward and reverse diffusion processes thus, speaking from the perspective of autoencoders, minimizing reconstruction error. Data-dependent priors have been proposed by \citet{Grad-TTS} and \citet{PriorGrad}, and we follow the former paper due to the flexibility of the continuous DPM framework used there. Our approach is summarized in Figure~\ref{pic:main}.

\begin{figure*}[ht]
\begin{center}
\center{\includegraphics[width=1.0\linewidth]{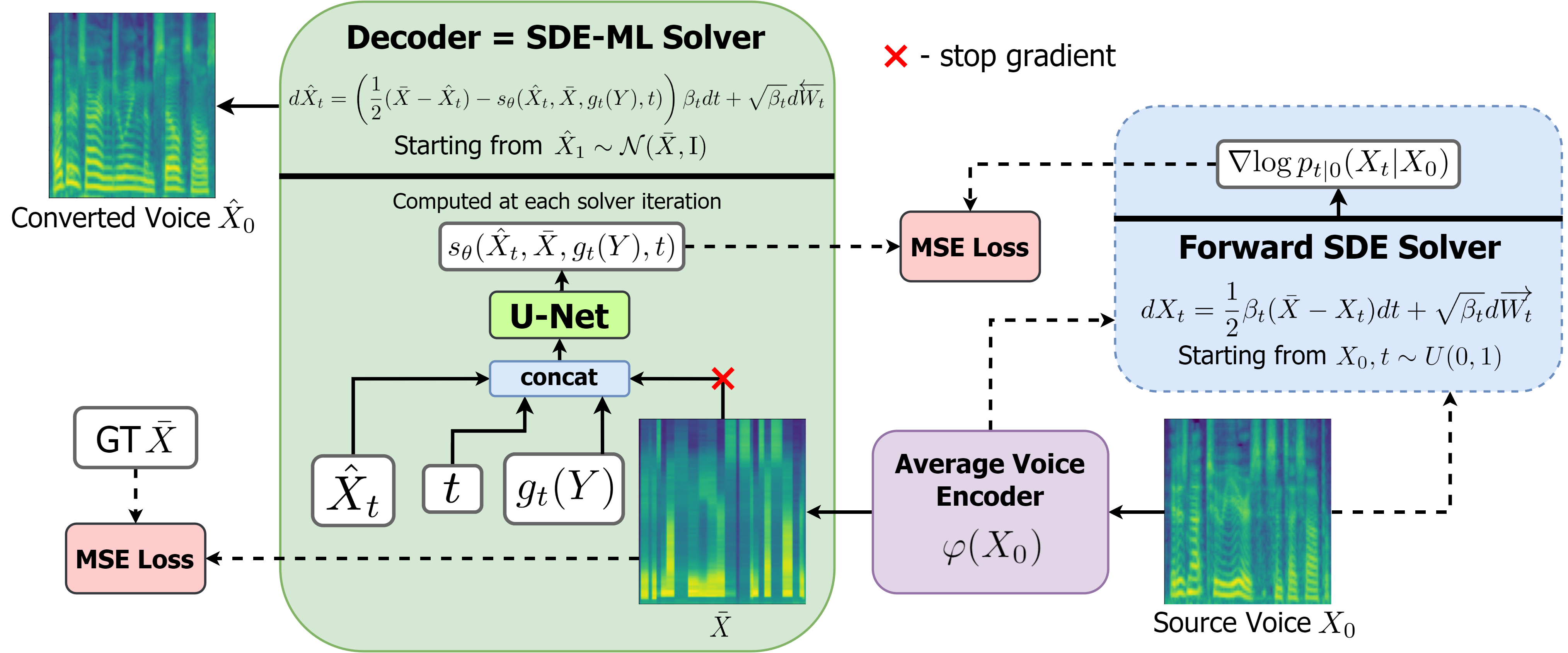}}
\end{center}
\caption{VC model training and inference. $Y$ stands for the training mel-spectrogram at training and the target mel-spectrogram at inference. Speaker conditioning in the decoder is enabled by the speaker conditioning network $g_{t}(Y)$ where $Y=\{Y_{t}\}_{t\in[0, 1]}$ is the whole forward diffusion trajectory starting at $Y_{0}$. Dotted arrows denote operations performed only at training.}
\label{pic:main}
\end{figure*}

\subsection{Encoder}
\label{subsec:encoder}

We choose average phoneme-level mel features as speaker-independent speech representation. To train the encoder to convert input mel-spectrograms into those of ``average voice'', we take three steps: (i) first, we apply Montreal Forced Aligner \citep{MFA} to a large-scale multi-speaker LibriTTS dataset \citep{LibriTTS} to align speech frames with phonemes;
(ii) next, we obtain average mel features for each particular phoneme by aggregating its mel features across the whole LibriTTS dataset;
(iii) the encoder is then trained to minimize mean square error between output mel-spectrograms and ground truth ``average voice'' mel-spectrograms (i.e. input mel-spectrograms where each phoneme mel feature is replaced with the average one calculated on the previous step).

The encoder has exactly the same Transformer-based architecture used in Grad-TTS \citep{Grad-TTS} except that its inputs are mel features rather than character or phoneme embeddings. Note that unlike Grad-TTS the encoder is trained separately from the decoder described in the next section.

\subsection{Decoder}
\label{subsec:decoder}

Whereas the encoder parameterizes the terminal distribution of the forward diffusion (i.e. the prior), the reverse diffusion is parameterized with the decoder. Following \citet{Song-main} we use It\^o calculus and define diffusions in terms of stochastic processes rather than discrete-time Markov chains.

The general DPM framework we utilize consists of forward and reverse diffusions given by the following Stochastic Differential Equations (SDEs):

\begin{equation}
\label{eq:vc_fwd_sde}
    dX_{t}=\frac{1}{2}\beta_{t}(\bar{X} - X_{t})dt + \sqrt{\beta_{t}}d\overrightarrow{W_{t}}\ ,
\end{equation}

\begin{equation}
\label{eq:vc_rev_sde}
    d\hat{X}_{t}=\left(\frac{1}{2}(\bar{X} - \hat{X}_{t}) - s_{\theta}(\hat{X}_{t}, \bar{X}, t)\right)\beta_{t}dt + \sqrt{\beta_{t}}d\overleftarrow{W_{t}}\ ,
\end{equation}

where $t\in [0,1]$, $\overrightarrow{W}$ and $\overleftarrow{W}$ are two independent Wiener processes in $\mathbb{R}^{n}$, $\beta_{t}$ is non-negative function referred to as \textit{noise schedule}, $s_{\theta}$ is the score function with parameters $\theta$ and $\bar{X}$ is $n$-dimensional vector. It can be shown \citep{Grad-TTS} that the forward SDE (\ref{eq:vc_fwd_sde}) allows for explicit solution:

\begin{equation}
\label{eq:vc_solution}
    \Law(X_{t}|X_{0}) = \mathcal{N}\left(e^{-\frac{1}{2}\int_{0}^{t}{\beta_{s}ds}}X_{0} + \left(1-e^{-\frac{1}{2}\int_{0}^{t}{\beta_{s}ds}}\right)\bar{X}, \left(1 - e^{-\int_{0}^{t}{\beta_{s}ds}}\right)\ide\right),
\end{equation}

where $\ide$ is $n\times n$ identity matrix. Thus, if noise follows linear schedule $\beta_{t}=\beta_{0} + t(\beta_{1}-\beta_{0})$ for $\beta_{0}$ and $\beta_{1}$ such that $e^{-\int_{0}^{1}{\beta_{s}ds}}$ is close to zero, then $\Law{(X_{1})}$ is close to $\mathcal{N}(\bar{X},\ide)$ which is the prior in this DPM. The reverse diffusion (\ref{eq:vc_rev_sde}) is trained by minimizing weighted $L_{2}$ loss:

\begin{equation}
\label{eq:vc_training}
    \theta^{*}=\argmin_{\theta}{\mathcal{L}(\theta)}=\argmin_{\theta}{\int_{0}^{1}{\lambda_{t}\mathbb{E}_{X_{0}, X_{t}}\Vert s_{\theta}(X_{t}, \bar{X}, t) - \nabla{\log{p_{t|0}(X_{t}|X_{0})}}\Vert_{2}^{2}}}dt,
\end{equation}

where $p_{t|0}(X_{t}|X_{0})$ is the probability density function (pdf) of the conditional distribution (\ref{eq:vc_solution}) and $\lambda_{t}=1 - e^{-\int_{0}^{t}{\beta_{s}ds}}$. The distribution (\ref{eq:vc_solution}) is Gaussian, so we have

\begin{equation}
\label{eq:vc_logp}
    \nabla\log{p_{t|0}(X_{t}|X_{0})} = -\frac{X_{t}-X_{0}e^{-\frac{1}{2}\int_{0}^{t}{\beta_{s}ds}}-\bar{X}(1-e^{-\frac{1}{2}\int_{0}^{t}{\beta_{s}ds}})}{1 - e^{-\int_{0}^{t}{\beta_{s}ds}}}.
\end{equation}

At training, time variable $t$ is sampled uniformly from $[0,1]$, noisy samples $X_{t}$ are generated according to the formula (\ref{eq:vc_solution}) and the formula (\ref{eq:vc_logp}) is used to calculate loss function $\mathcal{L}$ on these samples. Note that $X_{t}$ can be sampled without the necessity to calculate intermediate values $\{X_{s}\}_{0<s<t}$ which makes optimization task (\ref{eq:vc_training}) time and memory efficient. A well-trained reverse diffusion (\ref{eq:vc_rev_sde}) has trajectories that are close to those of the forward diffusion (\ref{eq:vc_fwd_sde}), so generating data with this DPM can be performed by sampling $\hat{X}_{1}$ from the prior $\mathcal{N}(\bar{X}, \ide)$ and solving SDE (\ref{eq:vc_rev_sde}) backwards in time.

The described above DPM was introduced by \citet{Grad-TTS} for text-to-speech task and we adapt it for our purposes. We put $\bar{X}=\varphi(X_{0})$ where $\varphi$ is the encoder, i.e. $\bar{X}$ is the ``average voice'' mel-spectrogram which we want to transform into that of the target voice. We condition the decoder $s_{\theta}=s_{\theta}(\hat{X}_{t},\bar{X},g_{t}(Y),t)$ on some trainable function $g_{t}(Y)$ to provide it with information about the target speaker ($Y$ stands for forward trajectories of the target mel-spectrogram at inference and the ones of the training mel-spectrogram at training). This function is a neural network trained jointly with the decoder. We experimented with three input types for this network:

\begin{itemize}
    \item \textit{d-only} -- the input is the speaker embedding extracted from the target mel-spectrogram $Y_{0}$ with the pre-trained speaker verification network employed in \citep{dvector};
    \item \textit{wodyn} -- in addition, the noisy target mel-spectrogram $Y_{t}$ is used as input;
    \item \textit{whole} -- in addition, the whole dynamics of the target mel-spectrogram under forward diffusion $\{Y_{s}|s=0.5/15,1.5/15,..,14.5/15\}$ is used as input.
\end{itemize}

The decoder architecture is based on U-Net \citep{UNet} and is the same as in Grad-TTS but with four times more channels to better capture the whole range of human voices. The speaker conditioning network $g_{t}(Y)$ is composed of $2$D convolutions and MLPs and described in detail in Appendix \ref{app:h}. Its output is $128$-dimensional vector which is broadcast-concatenated to the concatenation of $\hat{X}_{t}$ and $\bar{X}$ as additional $128$ channels.

\subsection{Related VC models}
\label{subsec:related}
To the best of our knowledge, there exist two diffusion-based voice conversion models: VoiceGrad \citep{VoiceGrad} and DiffSVC \citep{DiffSVC}. The one we propose differs from them in several important aspects. First, neither of the mentioned papers considers a one-shot many-to-many voice conversion scenario. Next, these models take no less than $100$ reverse diffusion steps at inference while we pay special attention to reducing the number of iterations (see Section \ref{sec:ml_inference}) achieving good quality with as few as $6$ iterations. Furthermore, VoiceGrad performs voice conversion by running Langevin dynamics starting from the source mel-spectrogram, thus implicitly assuming that forward diffusion trajectories starting from the mel-spectrogram we want to synthesize are likely to pass through the neighborhood of the source mel-spectrogram on their way to Gaussian noise. Such an assumption allowing to have only one network instead of encoder-decoder architecture is too strong and hardly holds for real voices. Finally, DiffSVC performs singing voice conversion and relies on PPGs as speaker-independent speech representation.

\section{Maximum likelihood SDE solver}
\label{sec:ml_inference}

In this section, we develop a fixed-step first-order reverse SDE solver that maximizes the log-likelihood of sample paths of the forward diffusion. This solver differs from general-purpose Euler-Maruyama SDE solver \citep{SDE-numerical} by infinitesimally small values which can however become significant when we sample from diffusion model using a few iterations.

Consider the following forward and reverse SDEs defined in Euclidean space $\mathbb{R}^{n}$ for $t\in [0,1]$:

\begin{equation}
\label{eq:fwd_rev_sde}
    dX_{t}=-\frac{1}{2}\beta_{t}X_{t}dt + \sqrt{\beta_{t}}d\overrightarrow{W_{t}} \ \ (F), \ \ \ d\hat{X}_{t}=\left(-\frac{1}{2}\beta_{t}\hat{X}_{t} - \beta_{t}s_{\theta}(\hat{X}_t, t)\right)dt + \sqrt{\beta_{t}}d\overleftarrow{W_{t}} \ \ (R),
\end{equation}

where $\overrightarrow{W}$ is a forward Wiener process (i.e. its forward increments $\overrightarrow{W_{t}} - \overrightarrow{W_{s}}$ are independent of $\overrightarrow{W_{s}}$ for $t > s$) and $\overleftarrow{W}$ is a backward Wiener process (i.e. backward increments $\overleftarrow{W_{s}} - \overleftarrow{W_{t}}$ are independent of $\overleftarrow{W_{t}}$ for $s < t$). Following \citet{Song-main} we will call DPM (\ref{eq:fwd_rev_sde}) Variance Preserving (VP). For simplicity we will derive maximum likelihood solver for this particular type of diffusion models. The equation (\ref{eq:vc_fwd_sde}) underlying VC diffusion model described in Section \ref{sec:vc_model} can be transformed into the equation (\ref{eq:fwd_rev_sde}-F) by a constant shift and we will call such diffusion models Mean Reverting Variance Preserving (MR-VP). VP model analysis carried out in this section can be easily extended (see Appendices \ref{app:mr_vp}, \ref{app:sub_vp} and \ref{app:ve}) to MR-VP model as well as to other common diffusion model types such as sub-VP and VE described by \citet{Song-main}.

The forward SDE (\ref{eq:fwd_rev_sde}-F) allows for explicit solution:

\begin{equation}
\label{eq:xt_distribution}
    \Law(X_{t} | X_{s}) = \mathcal{N}(\gamma_{s,t}X_{s}, (1-\gamma_{s,t}^2)\ide), \ \ \ \ \gamma_{s,t} = \exp{\left(-\frac{1}{2}\int_{s}^{t}{\beta_{u}du}\right)},
\end{equation}

for all $0\leq s < t \leq 1$. This formula is derived by means of It\^o calculus in Appendix \ref{app:sde_solution}. The reverse SDE (\ref{eq:fwd_rev_sde}-R) parameterized with a neural network $s_{\theta}$ is trained to approximate gradient of the log-density of noisy data $X_{t}$:

\begin{equation}
\label{eq:objective}
    \theta^{*}=\argmin_{\theta}{\int_{0}^{1}{\lambda_{t}\mathbb{E}_{X_{t}}\Vert s_{\theta}(X_{t}, t) - \nabla{\log{p_{t}(X_{t})}}\Vert_{2}^{2}}}dt,
\end{equation}

where the expectation is taken with respect to noisy data distribution $\Law(X_{t})$ with pdf $p_{t}(\cdot)$ and $\lambda_{t}$ is some positive weighting function. Note that certain Lipschitz constraints should be satisfied by coefficients of SDEs (\ref{eq:fwd_rev_sde}) to guarantee existence of strong solutions \citep{SDE-book}, and throughout this section we assume these conditions are satisfied as well as those from \citep{SDE-reverse} which guarantee that paths $\hat{X}$ generated by the reverse SDE (\ref{eq:fwd_rev_sde}-R) for the optimal $\theta^{*}$ equal forward SDE (\ref{eq:fwd_rev_sde}-F) paths $X$ in distribution.

The generative procedure of a VP DPM consists in solving the reverse SDE (\ref{eq:fwd_rev_sde}-R) backwards in time starting from $\hat{X}_{1}\sim \mathcal{N}(0,\ide)$. Common Euler-Maruyama solver introduces discretization error \citep{SDE-numerical} which may harm sample quality when the number of iterations is small. At the same time, it is possible to design unbiased \citep{SDE-unbiased} or even exact \citep{SDE-exact} numerical solvers for some particular SDE types. The Theorem~\ref{th:main} shows that in the case of diffusion models we can make use of the forward diffusion (\ref{eq:fwd_rev_sde}-F) and propose a reverse SDE solver which is better than the general-purpose Euler-Maruyama one in terms of likelihood.

The solver proposed in the Theorem \ref{th:main} is expressed in terms of the values defined as follows:

\begin{equation}
\label{eq:notation_gamma}
    \mu_{s,t}=\gamma_{s,t}\frac{1-\gamma_{0,s}^{2}}{1-\gamma_{0,t}^{2}}, \ \ \ \  \nu_{s,t}=\gamma_{0,s}\frac{1-\gamma_{s,t}^{2}}{1-\gamma_{0,t}^{2}}, \ \ \ \
    \sigma_{s,t}^{2}=\frac{(1-\gamma_{0,s}^{2})(1-\gamma_{s,t}^{2})}{1-\gamma_{0,t}^{2}},
\end{equation}

\begin{equation}
\begin{split}
\label{eq:notation_correction}
     \kappa_{t,h}^{*} =& \frac{\nu_{t-h,t}(1-\gamma_{0,t}^{2})}{\gamma_{0,t}\beta_{t}h} - 1,  \ \ \ \ 
     \omega_{t,h}^{*}=\frac{\mu_{t-h,t}-1}{\beta_{t}h} + \frac{1+\kappa_{t,h}^{*}}{1-\gamma_{0,t}^{2}} -\frac{1}{2},
     \\& (\sigma^{*}_{t,h})^{2} = \sigma^{2}_{t-h,t} + \frac{1}{n}\nu_{t-h,t}^{2}\mathbb{E}_{X_{t}}\left[\Tr{\left(\Var{(X_{0}|X_{t})}\right)}\right],
\end{split}
\end{equation}

where $n$ is data dimensionality, $\Var{(X_{0}|X_{t})}$ is the covariance matrix of the conditional data distribution $\Law(X_{0}|X_{t})$ (so, $\Tr(\Var{(X_{0}|X_{t})})$ is the overall variance across all $n$ dimensions) and the expectation $\mathbb{E}_{X_{t}}[\cdot]$ is taken with respect to the unconditional noisy data distribution $\Law(X_{t})$.

\begin{theorem}
\label{th:main}
    Consider a DPM characterized by SDEs (\ref{eq:fwd_rev_sde}) with reverse diffusion trained till optimality. Let $N\in \mathbb{N}$ be any natural number and $h=1/N$. Consider the following class of fixed step size $h$ reverse SDE solvers parameterized with triplets of real numbers \{$(\hat{\kappa}_{t,h}, \hat{\omega}_{t,h}, \hat{\sigma}_{t,h})|t=h,2h,..,1$\}:
    \begin{equation}
    \label{eq:sde_solver}
        \hat{X}_{t-h} = \hat{X}_{t} + \beta_{t}h\left(\left(\frac{1}{2}+\hat{\omega}_{t,h}\right)\hat{X}_{t} + (1+\hat{\kappa}_{t,h})s_{\theta^{*}}(\hat{X}_{t},t)\right) + \hat{\sigma}_{t,h}\xi_{t},
    \end{equation}
    where $\theta^{*}$ is given by (\ref{eq:objective}), $t=1,1-h,..,h$ and $\xi_{t}$ are i.i.d. samples from $\mathcal{N}(0,\ide)$. Then:

    \begin{enumerate}
        \item[(i)] Log-likelihood of sample paths $X=\{X_{kh}\}_{k=0}^{N}$ under generative model $\hat{X}$ is maximized for $\hat{\kappa}_{t,h}=\kappa_{t,h}^{*}$, $\hat{\omega}_{t,h}=\omega_{t,h}^{*}$ and $\hat{\sigma}_{t,h}=\sigma^{*}_{t,h}$. 
        \item[(ii)] Assume that the SDE solver (\ref{eq:sde_solver}) starts from random variable $\hat{X}_{1}\sim \Law{(X_{1})}$. If $X_{0}$ is a constant or a Gaussian random variable with diagonal isotropic covariance matrix (i.e. $\delta^{2}\ide$ for $\delta>0$), then generative model $\hat{X}$ is exact for $\hat{\kappa}_{t,h}=\kappa_{t,h}^{*}$, $\hat{\omega}_{t,h}=\omega_{t,h}^{*}$ and $\hat{\sigma}_{t,h}=\sigma^{*}_{t,h}$.
    \end{enumerate}
\end{theorem}

The Theorem \ref{th:main} provides an improved DPM sampling scheme which comes at no additional computational cost compared to standard methods (except for data-dependent term in $\sigma^{*}$ as discussed in Section \ref{subsec:sampling}) and requires neither model re-training nor extensive search on noise schedule space. The proof of this theorem is given in Appendix \ref{app:th_proof}. Note that it establishes optimality of the reverse SDE solver (\ref{eq:sde_solver}) with the parameters ($\ref{eq:notation_correction}$) in terms of likelihood of \textit{discrete} paths $X=\{X_{kh}\}_{k=0}^{N}$ while the optimality of \textit{continuous} model (\ref{eq:fwd_rev_sde}-R) on \textit{continuous} paths $\{X_{t}\}_{t\in[0,1]}$ is guaranteed for a model with parameters $\theta=\theta^{*}$ as shown in \citep{Song-main}.

The class of reverse SDE solvers considered in the Theorem \ref{th:main} is rather broad: it is the class of all fixed-step solvers whose increments at time $t$ are linear combination of $\hat{X}_{t}$, $s_{\theta}(\hat{X}_{t},t)$ and Gaussian noise with zero mean and diagonal isotropic covariance matrix. As a particular case it includes Euler-Maruyama solver ($\hat{\kappa}_{t,h}\equiv 0$, $\hat{\omega}_{t,h}\equiv 0$, $\hat{\sigma}_{t,h}\equiv\sqrt{\beta_{t}h}$) and for fixed $t$ and $h\to 0$ we have $\kappa^{*}_{t,h}=\bar{o}(1)$, $\omega^{*}_{t,h}=\bar{o}(1)$ and $\sigma^{*}_{t,h}=\sqrt{\beta_{t}h}(1 + \bar{o}(1))$ (the proof is given in Appendix \ref{app:asymptotics}), so the optimal SDE solver significantly differs from general-purpose Euler-Maruyama solver only when $N$ is rather small or $t$ has the same order as $h$, i.e. on the final steps of DPM inference. Appendix \ref{app:toy} contains toy examples demonstrating the difference of the proposed optimal SDE solver and Euler-Maruyama one depending on step size.

The result \textit{(ii)} from the Theorem \ref{th:main} strengthens the result \textit{(i)} for some particular data distributions, but it may seem useless since in practice data distribution is far from being constant or Gaussian. However, in case of generation with strong conditioning (e.g. mel-spectrogram inversion) the assumptions on the data distribution may become viable: in the limiting case when our model is conditioned on $c=\psi(X_{0})$ for an injective function $\psi$, random variable $X_{0}|c$ becomes a constant $\psi^{-1}(c)$.

\section{Experiments}
\label{sec:exp}
We trained two groups of models: \textit{Diff-VCTK} models on VCTK \citep{VCTK} dataset containing $109$ speakers ($9$ speakers were held out for testing purposes) and \textit{Diff-LibriTTS} models on LibriTTS \citep{LibriTTS} containing approximately $1100$ speakers ($10$ speakers were held out). For every model both encoder and decoder were trained on the same dataset. Training hyperparameters, implementation and data processing details can be found in Appendix \ref{app:details}. For mel-spectrogram inversion, we used the pre-trained universal HiFi-GAN vocoder \citep{HiFi-GAN} operating at $22.05$kHz. All subjective human evaluation was carried out on Amazon Mechanical Turk (AMT) with Master assessors to ensure the reliability of the obtained Mean Opinion Scores (MOS). In all AMT tests we considered unseen-to-unseen conversion with $25$ unseen (for both \textit{Diff-VCTK} and \textit{Diff-LibriTTS}) speakers: $9$ VCTK speakers, $10$ LibriTTS speakers and $6$ internal speakers. For VCTK source speakers we also ensured that source phrases were unseen during training. We place other details of listening AMT tests in Appendix \ref{app:subj}. A small subset of speech samples used in them is available at our demo page \url{https://diffvc-fast-ml-solver.github.io} which we encourage to visit.

\begin{table}
\caption{Input types for speaker conditioning $g_{t}(Y)$ compared in terms of speaker similarity.}
\begin{center}
\begin{tabular}{|c|c|c|c|c|c|c|}
\hline
&
\multicolumn{3}{c|}{\begin{tabular}[c]{@{}c@{}c@{}}\textit{Diff-LibriTTS}\end{tabular}} &
\multicolumn{3}{c|}{\begin{tabular}[c]{@{}c@{}c@{}}\textit{Diff-VCTK}\end{tabular}} \\ \cline{2-7}
&\textit{d-only} &\textit{wodyn} &\textit{whole} &\textit{d-only} &\textit{wodyn} &\textit{whole} \\ \hline
Most similar &$27.0\%$ &$\mathbf{38.0\%}$ &$34.1\%$ &$27.2\%$ &$\mathbf{46.7\%}$ &$23.6\%$\\ \hline
Least similar &$\mathbf{28.9\%}$ &$29.3\%$ &$38.5\%$ &$25.3\%$ &$\mathbf{23.9\%}$ &$48.6\%$\\ \hline
\end{tabular}
\end{center}
\label{tab:conditioning}
\end{table}

As for sampling, we considered the following class of reverse SDE solvers:
\begin{equation}
\label{eq:vc_solver}
    \hat{X}_{t-h} = \hat{X}_{t} + \beta_{t}h\left(\left(\frac{1}{2}+\hat{\omega}_{t,h}\right)(\hat{X}_{t} - \bar{X}) + (1+\hat{\kappa}_{t,h})s_{\theta}(\hat{X}_{t},\bar{X},g_{t}(Y),t)\right) + \hat{\sigma}_{t,h}\xi_{t},
\end{equation}

where $t=1,1-h,..,h$ and $\xi_{t}$ are i.i.d. samples from $\mathcal{N}(0,\ide)$. For $\hat{\kappa}_{t,h}=\kappa^{*}_{t,h}$, $\hat{\omega}_{t,h}=\omega^{*}_{t,h}$ and $\hat{\sigma}_{t,h}=\sigma^{*}_{t,h}$ (where $\kappa^{*}_{t,h}$, $\omega^{*}_{t,h}$ and $\sigma^{*}_{t,h}$ are given by (\ref{eq:notation_correction})) it becomes maximum likelihood reverse SDE solver for MR-VP DPM (\ref{eq:vc_fwd_sde}-\ref{eq:vc_rev_sde}) as shown in Appendix \ref{app:mr_vp}. In practice it is not trivial to estimate variance of the conditional distribution $\Law{(X_{0}|X_{t})}$, so we skipped this term in $\sigma^{*}_{t,h}$ assuming this variance to be rather small because of strong conditioning on $g_{t}(Y)$ and just used $\hat{\sigma}_{t,h}=\sigma_{t-h,t}$ calling this sampling method \textit{ML-N} ($N=1/h$ is the number of SDE solver steps). We also experimented with Euler-Maruyama solver \textit{EM-N} (i.e. $\hat{\kappa}_{t,h}=0$, $ \hat{\omega}_{t,h}=0$, $\hat{\sigma}_{t,h}=\sqrt{\beta_{t}h}$) and ``probability flow sampling'' from \citep{Song-main} which we denote by \textit{PF-N} ($\hat{\kappa}_{t,h}=-0.5$, $ \hat{\omega}_{t,h}=0$, $\hat{\sigma}_{t,h}=0$).

\subsection{Speaker conditioning analysis}
\label{subsec:ablation}

For each dataset we trained three models -- one for each input type for the speaker conditioning network $g_{t}(Y)$ (see Section \ref{subsec:decoder}). Although these input types had much influence neither on speaker similarity nor on speech naturalness, we did two experiments to choose the best models (one for each training dataset) in terms of speaker similarity for further comparison with baseline systems. We compared voice conversion results (produced by \textit{ML-30} sampling scheme) on $92$ source-target pairs. AMT workers were asked which of three models (if any) sounded most similar to the target speaker and which of them (if any) sounded least similar. For \textit{Diff-VCTK} and \textit{Diff-LibriTTS} models each conversion pair was evaluated $4$ and $5$ times respectively. Table~\ref{tab:conditioning} demonstrates that for both  \textit{Diff-VCTK} and \textit{Diff-LibriTTS} the best option is \textit{wodyn}, i.e. to condition the decoder at time $t$ on the speaker embedding together with the noisy target mel-spectrogram $Y_{t}$. Conditioning on $Y_{t}$ allows making use of diffusion-specific information of how the noisy target sounds whereas embedding from the pre-trained speaker verification network contains information only about the clean target. Taking these results into consideration, we used \textit{Diff-VCTK-wodyn} and \textit{Diff-LibriTTS-wodyn} in the remaining experiments.

\subsection{Any-to-any voice conversion}
\label{subsec:any2any}

\begin{table}
\caption{Subjective evaluation (MOS) of one-shot VC models trained on VCTK. Ground truth recordings were evaluated only for VCTK speakers.}
\begin{center}
\begin{tabular}{|c|c|c|c|c|}
\hline
&
\multicolumn{2}{c|}{\begin{tabular}[c]{@{}c@{}}VCTK test ($9$ speakers, $54$ pairs)\end{tabular}} &
\multicolumn{2}{c|}{\begin{tabular}[c]{@{}c@{}}Whole test ($25$ speakers, $350$ pairs)\end{tabular}} \\ \cline{2-5}
&\ \ \ Naturalness \ \ \ &Similarity &\ \ \ Naturalness \ \ \ &Similarity \\ \hline
\textit{AGAIN-VC} &$1.98\pm 0.05$ &$1.97\pm 0.08$ &$1.87\pm 0.03$ &$1.75\pm 0.04$\\ \hline
\textit{FragmentVC} &$2.20\pm 0.06$ &$2.45\pm 0.09$ &$1.91\pm 0.03$ &$1.93\pm 0.04$\\ \hline
\textit{VQMIVC} &$2.89\pm 0.06$ &$2.60\pm 0.10$ &$2.48\pm 0.04$ &$1.95\pm 0.04$\\ \hline
\textit{Diff-VCTK-ML-6} &$\mathbf{3.73\pm 0.06}$ &$3.47\pm 0.09$ &$3.39\pm 0.04$ &$\mathbf{2.69\pm 0.05}$\\ \hline
\textit{Diff-VCTK-ML-30} &$\mathbf{3.73\pm 0.06}$ &$\mathbf{3.57\pm 0.09}$ &$\mathbf{3.44\pm 0.04}$ &$\mathbf{2.71\pm 0.05}$\\ \hline
\textit{Ground truth} &$4.55\pm 0.05$ &$4.52\pm 0.07$ &$4.55\pm 0.05$ &$4.52\pm 0.07$\\ \hline
\end{tabular}
\end{center}
\label{tab:main_vctk}
\end{table}

\begin{table}
\caption{Subjective evaluation (MOS) of one-shot VC models trained on large-scale datasets.}
\begin{center}
\begin{tabular}{|c|c|c|c|c|}
\hline
&
\multicolumn{2}{c|}{\begin{tabular}[c]{@{}c@{}}VCTK test ($9$ speakers, $54$ pairs)\end{tabular}} &
\multicolumn{2}{c|}{\begin{tabular}[c]{@{}c@{}}Whole test ($25$ speakers, $350$ pairs)\end{tabular}} \\ \cline{2-5}
&\ \ \ Naturalness \ \ \ \  &Similarity &\ \ \ Naturalness \ \ \ \ &Similarity \\ \hline
\textit{Diff-LibriTTS-EM-6} &$1.68\pm 0.06$ &$1.53\pm 0.07$ &$1.57\pm 0.02$ &$1.47\pm 0.03$\\ \hline
\textit{Diff-LibriTTS-PF-6} &$3.11\pm 0.07$ &$2.58\pm 0.11$ &$2.99\pm 0.03$ &$2.50\pm 0.04$\\ \hline
\textit{Diff-LibriTTS-ML-6} &$3.84\pm 0.08$ &$3.08\pm 0.11$ &$3.80\pm 0.03$ &$3.27\pm 0.05$\\ \hline
\textit{Diff-LibriTTS-ML-30} &$\mathbf{3.96\pm 0.08}$ &$3.23\pm 0.11$ &$\mathbf{4.02\pm 0.03}$ &$\mathbf{3.39\pm 0.05}$\\ \hline
\textit{BNE-PPG-VC} &$\mathbf{3.95\pm 0.08}$ &$\mathbf{3.27\pm 0.12}$ &$3.83\pm 0.03$ &$3.03\pm 0.05$\\ \hline
\end{tabular}
\end{center}
\label{tab:main_large}
\end{table}

We chose four recently proposed VC models capable of one-shot many-to-many synthesis as the baselines: 
\begin{itemize}
    \item \textit{AGAIN-VC} \citep{AGAIN}, an improved version of a conventional autoencoder AdaIN-VC solving the disentanglement problem by means of instance normalization;
    \item \textit{FragmentVC} \citep{FragmentVC}, an attention-based model relying on wav2vec 2.0 \citep{Wav2Vec} to obtain speech content from the source utterance;
    \item \textit{VQMIVC} \citep{VQMIVC}, state-of-the-art approach among those employing vector quantization techniques;
    \item \textit{BNE-PPG-VC} \citep{PPG-VC}, an improved variant of PPG-based VC models combining a bottleneck feature extractor obtained from a phoneme recognizer with a seq2seq-based synthesis module.
\end{itemize}
    
As shown in \citep{Assem-VC}, PPG-based VC models provide high voice conversion quality competitive even with that of the state-of-the-art VC models taking text transcription corresponding to the source utterance as input. Therefore, we can consider \textit{BNE-PPG-VC} a state-of-the-art model in our setting.

Baseline voice conversion results were produced by the pre-trained VC models provided in official GitHub repositories. Since only \textit{BNE-PPG-VC} has the model pre-trained on a large-scale dataset (namely, LibriTTS + VCTK), we did two subjective human evaluation tests: the first one comparing \textit{Diff-VCTK} with \textit{AGAIN-VC}, \textit{FragmentVC} and \textit{VQMIVC} trained on VCTK and the second one comparing \textit{Diff-LibriTTS} with \textit{BNE-PPG-VC}. The results of these tests are given in Tables \ref{tab:main_vctk} and \ref{tab:main_large} respectively. Speech naturalness and speaker similarity were assessed separately. AMT workers evaluated voice conversion quality on $350$ source-target pairs on $5$-point scale. In the first test, each pair was assessed $6$ times on average both in speech naturalness and speaker similarity evaluation; as for the second one, each pair was assessed $8$ and $9$ times on average in speech naturalness and speaker similarity evaluation correspondingly. No less than $41$ unique assessors took part in each test.

Table~\ref{tab:main_vctk} demonstrates that our model performs significantly better than the baselines both in terms of naturalness and speaker similarity even when $6$ reverse diffusion iterations are used. Despite working almost equally well on VCTK speakers, the best baseline \textit{VQMIVC} shows poor performance on other speakers perhaps because of not being able to generalize to different domains with lower recording quality. Although \textit{Diff-VCTK} performance also degrades on non-VCTK speakers, it achieves good speaker similarity of MOS $3.6$ on VCTK ones when \textit{ML-30} sampling scheme is used and only slightly worse MOS $3.5$ when $5$x less iterations are used at inference.

\begin{table}
\caption{Reverse SDE solvers compared in terms of FID. $N$ is the number of SDE solver steps.}
\begin{center}
\begin{tabular}{|c|c|c|c|c|c|c|}
\hline
&
\multicolumn{2}{c|}{\begin{tabular}[c]{@{}c@{}}VP DPM\end{tabular}} &
\multicolumn{2}{c|}{\begin{tabular}[c]{@{}c@{}}sub-VP DPM\end{tabular}} &
\multicolumn{2}{c|}{\begin{tabular}[c]{@{}c@{}}VE DPM\end{tabular}} \\ \cline{2-7}
&$N$=$10$ &$N$=$100$ &$N$=$10$ &$N$=$100$ &$N$=$10$ &$N$=$100$ \\ \hline
Euler-Maruyama &$229.6$ &$19.68$ &$312.3$ &$19.83$ &$462.1$ &$24.77$\\ \hline
Reverse Diffusion &$679.8$ &$65.95$ &$312.2$ &$19.74$ &$461.1$ &$303.2$\\ \hline
Probability Flow &$88.92$ &$5.70$ &$64.22$ &$\mathbf{4.42}$ &$495.3$ &$214.5$\\ \hline
Ancestral Sampling &$679.8$ &$68.35$ &--- &--- &$454.7$ &$17.83$ \\ \hline
Maximum Likelihood ($\tau=0.1$) &$260.3$ &$\mathbf{4.34}$ &$317.0$ &$6.63$ &$461.9$ &$23.63$\\ \hline
Maximum Likelihood ($\tau=0.5$) &$\mathbf{24.45}$ &$7.82$ &$\mathbf{30.90}$ &$6.43$ &$462.0$ &$\mathbf{10.07}$\\ \hline
Maximum Likelihood ($\tau=1.0$) &$41.78$ &$7.94$ &$48.02$ &$6.51$ &$\mathbf{48.51}$ &$12.37$\\ \hline
\end{tabular}
\end{center}
\label{tab:sampling}
\end{table}

Table~\ref{tab:main_large} contains human evaluation results of \textit{Diff-LibriTTS} for four sampling schemes: \textit{ML-30} with $30$ reverse SDE solver steps and \textit{ML-6}, \textit{EM-6} and \textit{PF-6} with $6$ steps of reverse diffusion. The three schemes taking $6$ steps achieved real-time factor (RTF) around $0.1$ on GPU (i.e. inference was $10$ times faster than real time) while the one taking $30$ steps had RTF around $0.5$. The proposed model \textit{Diff-LibriTTS-ML-30} and the baseline \textit{BNE-PPG-VC} show the same performance on the VCTK test set in terms of speech naturalness the latter being slightly better in terms of speaker similarity which can perhaps be explained by the fact that \textit{BNE-PPG-VC} was trained on the union of VCTK and LibriTTS whereas our model was trained only on LibriTTS. As for the whole test set containing unseen LibriTTS and internal speakers also, \textit{Diff-LibriTTS-ML-30} outperforms \textit{BNE-PPG-VC} model achieving MOS $4.0$ and $3.4$ in terms of speech naturalness and speaker similarity respectively. Due to employing PPG extractor trained on a large-scale ASR dataset LibriSpeech \citep{LibriSpeech}, \textit{BNE-PPG-VC} has fewer mispronunciation issues than our model but synthesized speech suffers from more sonic artifacts. This observation makes us believe that incorporating PPG features in the proposed diffusion VC framework is a promising direction for future research.

Table~\ref{tab:main_large} also demonstrates the benefits of the proposed maximum likelihood sampling scheme over other sampling methods for a small number of inference steps: only \textit{ML-N} scheme allows us to use as few as $N=6$ iterations with acceptable quality degradation of MOS $0.2$ and $0.1$ in terms of naturalness and speaker similarity respectively while two other competing methods lead to much more significant quality degradation.

\subsection{Maximum likelihood sampling}
\label{subsec:sampling}

\begin{figure*}[ht!]
    \includegraphics[width=.23\textwidth]{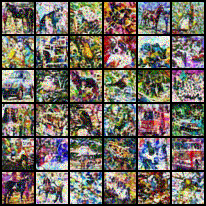}\hfill
    \includegraphics[width=.23\textwidth]{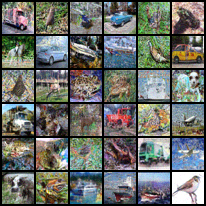}\hfill
    \includegraphics[width=.23\textwidth]{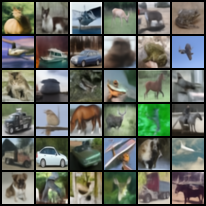}\hfill
    \includegraphics[width=.23\textwidth]{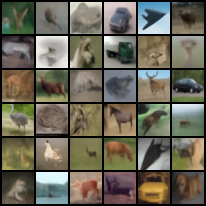}
\caption{CIFAR-$10$ images randomly sampled from VP DPM by running $10$ reverse diffusion steps with the following schemes (from left to right): ``euler-maruyama'', ``probability flow'', ``maximum likelihood ($\tau=0.5$)'', ``maximum likelihood ($\tau=1.0$)''.}
\label{pic:cifar}
\end{figure*}

To show that the maximum likelihood sampling scheme proposed in Section \ref{sec:ml_inference} generalizes to different tasks and DPM types, we took the models trained by \citet{Song-main} on CIFAR-$10$ image generation task and compared our method with other sampling schemes described in that paper in terms of Fr\'{e}chet Inception Distance (FID).

The main difficulty in applying maximum likelihood SDE solver is estimating data-dependent term $\mathbb{E}[\Tr{(\Var{(X_{0}|X_{t})})}]$ in $\sigma^{*}_{t,h}$. Although in the current experiments we just set this term to zero, we can think of two possible ways to estimate it: (i) approximate $\Var{(X_{0}|X_{t})}$ with $\Var{(\hat{X}_{0}|\hat{X}_{t}=X_{t})}$: sample noisy data $X_{t}$, solve reverse SDE with sufficiently small step size starting from terminal condition $\hat{X}_{t}=X_{t}$ several times, and calculate sample variance of the resulting solutions at initial points $\hat{X}_{0}$; (ii) use the formula (\ref{eq:gaussian_revert}) from Appendix \ref{app:th_proof} to calculate $\Var{(X_{0}|X_{t})}$ assuming that $X_{0}$ is distributed normally with mean and variance equal to sample mean and sample variance computed on the training dataset. Experimenting with these techniques and exploring new ones seems to be an interesting future research direction.

Another important practical consideration is that the proposed scheme is proven to be optimal only for score matching networks trained till optimality. Therefore, in the experiments whose results are reported in Table~\ref{tab:sampling} we apply maximum likelihood sampling scheme only when $t\leq\tau$ while using standard Euler-Maruyama solver for $t>\tau$ for some hyperparameter $\tau\in [0,1]$. Such a modification relies on the assumption that score matching network is closer to being optimal for smaller noise. 

Table \ref{tab:sampling} shows that despite likelihood and FID are two metrics that do not perfectly correlate \citep{DPM-ML-training}, in most cases our maximum likelihood SDE solver performs best in terms of FID. Also, it is worth mentioning that although $\tau=1$ is always rather a good choice, tuning this hyperparameter can lead to even better performance. One can find randomly chosen generated images for various sampling methods in Figure~\ref{pic:cifar}.

\section{Conclusion}
\label{sec:outro}
In this paper, the novel one-shot many-to-many voice conversion model has been presented. Its encoder design and powerful diffusion-based decoder made it possible to achieve good results both in terms of speaker similarity and speech naturalness even on out-of-domain unseen speakers. Subjective human evaluation verified that the proposed model delivers scalable VC solution with competitive performance. Furthermore, aiming at fast synthesis, we have developed and theoretically justified the novel sampling scheme. The main idea behind it is to modify the general-purpose Euler-Maruyama SDE solver so as to maximize the likelihood of discrete sample paths of the forward diffusion. Due to the proposed sampling scheme, our VC model is capable of high-quality voice conversion with as few as $6$ reverse diffusion steps. Moreover, experiments on the image generation task show that all known diffusion model types can benefit from the proposed SDE solver.

\bibliography{diffusion_paper}

\begin{thebibliography}{44}
\providecommand{\natexlab}[1]{#1}
\providecommand{\url}[1]{\texttt{#1}}
\expandafter\ifx\csname urlstyle\endcsname\relax
  \providecommand{\doi}[1]{doi: #1}\else
  \providecommand{\doi}{doi: \begingroup \urlstyle{rm}\Url}\fi

\bibitem[Anderson(1982)]{SDE-reverse}
Brian~D.O. Anderson.
\newblock {Reverse-time Diffusion Equation Models}.
\newblock \emph{Stochastic Processes and their Applications}, 12\penalty0
  (3):\penalty0 313 -- 326, 1982.
\newblock ISSN 0304-4149.

\bibitem[Baevski et~al.(2020)Baevski, Zhou, Mohamed, and Auli]{Wav2Vec}
Alexei Baevski, Yuhao Zhou, Abdelrahman Mohamed, and Michael Auli.
\newblock {wav2vec 2.0: A Framework for Self-Supervised Learning of Speech
  Representations}.
\newblock In \emph{Advances in Neural Information Processing Systems},
  volume~33, pp.\  12449--12460. Curran Associates, Inc., 2020.

\bibitem[Beskos \& Roberts(2005)Beskos and Roberts]{SDE-exact}
Alexandros Beskos and Gareth~O. Roberts.
\newblock {Exact Simulation of Diffusions}.
\newblock \emph{The Annals of Applied Probability}, 15\penalty0 (4):\penalty0
  2422--2444, 2005.
\newblock ISSN 10505164.

\bibitem[Chen et~al.(2021{\natexlab{a}})Chen, Zhang, Zen, Weiss, Norouzi, and
  Chan]{WaveGrad}
Nanxin Chen, Yu~Zhang, Heiga Zen, Ron~J Weiss, Mohammad Norouzi, and William
  Chan.
\newblock {WaveGrad: Estimating Gradients for Waveform Generation}.
\newblock In \emph{International Conference on Learning Representations},
  2021{\natexlab{a}}.

\bibitem[Chen et~al.(2021{\natexlab{b}})Chen, Wu, Wu, and yi~Lee]{AGAIN}
Yen-Hao Chen, D.~Wu, Tsung-Han Wu, and Hung yi~Lee.
\newblock {Again-VC: A One-Shot Voice Conversion Using Activation Guidance and
  Adaptive Instance Normalization}.
\newblock In \emph{ICASSP 2021 - 2021 IEEE International Conference on
  Acoustics, Speech and Signal Processing (ICASSP)}, pp.\  5954--5958,
  2021{\natexlab{b}}.

\bibitem[Chou \& Lee(2019)Chou and Lee]{AdaIN}
Ju{-}Chieh Chou and Hung{-}yi Lee.
\newblock {One-Shot Voice Conversion by Separating Speaker and Content
  Representations with Instance Normalization}.
\newblock In \emph{Interspeech 2019, 20th Annual Conference of the
  International Speech Communication Association, Graz, Austria, 15-19
  September 2019}, pp.\  664--668. {ISCA}, 2019.

\bibitem[Henry-Labordère et~al.(2017)Henry-Labordère, Tan, and
  Touzi]{SDE-unbiased}
Pierre Henry-Labordère, Xiaolu Tan, and Nizar Touzi.
\newblock {Unbiased Simulation of Stochastic Differential Equations}.
\newblock \emph{The Annals of Applied Probability}, 27\penalty0 (6):\penalty0
  3305--3341, 2017.
\newblock ISSN 10505164.

\bibitem[Ho et~al.(2020)Ho, Jain, and Abbeel]{DDPM}
Jonathan Ho, Ajay Jain, and Pieter Abbeel.
\newblock {Denoising Diffusion Probabilistic Models}.
\newblock In \emph{Advances in Neural Information Processing Systems 33: Annual
  Conference on Neural Information Processing Systems 2020, NeurIPS 2020,
  December 6-12, 2020, virtual}, volume~33. Curran Associates, Inc., 2020.

\bibitem[Hyv{{\"a}}rinen(2005)]{Score-Matching1}
Aapo Hyv{{\"a}}rinen.
\newblock {Estimation of Non-Normalized Statistical Models by Score Matching}.
\newblock \emph{Journal of Machine Learning Research}, 6\penalty0
  (24):\penalty0 695--709, 2005.

\bibitem[Ishihara \& Saito(2020)Ishihara and Saito]{AttentionEmbeds}
Tatsuma Ishihara and Daisuke Saito.
\newblock {Attention-Based Speaker Embeddings for One-Shot Voice Conversion}.
\newblock In \emph{Interspeech 2020, 21st Annual Conference of the
  International Speech Communication Association, Virtual Event, Shanghai,
  China, 25-29 October 2020}, pp.\  806--810. {ISCA}, 2020.

\bibitem[Jeong et~al.(2021)Jeong, Kim, Cheon, Choi, and Kim]{DiffTTS}
Myeonghun Jeong, Hyeongju Kim, Sung~Jun Cheon, Byoung~Jin Choi, and Nam~Soo
  Kim.
\newblock {Diff-TTS: A Denoising Diffusion Model for Text-to-Speech}.
\newblock In \emph{Proc. Interspeech 2021}, pp.\  3605--3609, 2021.

\bibitem[Jia et~al.(2018)Jia, Zhang, Weiss, Wang, Shen, Ren, Chen, Nguyen,
  Pang, Lopez~Moreno, and Wu]{dvector}
Ye~Jia, Yu~Zhang, Ron Weiss, Quan Wang, Jonathan Shen, Fei Ren, Zhifeng Chen,
  Patrick Nguyen, Ruoming Pang, Ignacio Lopez~Moreno, and Yonghui Wu.
\newblock {Transfer Learning from Speaker Verification to Multispeaker
  Text-To-Speech Synthesis}.
\newblock In \emph{Advances in Neural Information Processing Systems 31}, pp.\
  4480--4490. Curran Associates, Inc., 2018.

\bibitem[Kameoka et~al.(2020)Kameoka, Kaneko, Tanaka, Hojo, and
  Seki]{VoiceGrad}
Hirokazu Kameoka, Takuhiro Kaneko, Kou Tanaka, Nobukatsu Hojo, and Shogo Seki.
\newblock {VoiceGrad: Non-Parallel Any-to-Many Voice Conversion with Annealed
  Langevin Dynamics}, 2020.

\bibitem[Kim et~al.(2021)Kim, Park, and Joe]{Assem-VC}
{Kang-wook} Kim, {Seung-won} Park, and {Myun-chul} Joe.
\newblock {Assem-VC: Realistic Voice Conversion by Assembling Modern Speech
  Synthesis Techniques}, 2021.

\bibitem[Kingma et~al.(2021)Kingma, Salimans, Poole, and Ho]{VDPM}
Diederik~P. Kingma, Tim Salimans, Ben Poole, and Jonathan Ho.
\newblock {Variational Diffusion Models}, 2021.

\bibitem[Kloeden \& Platen(1992)Kloeden and Platen]{SDE-numerical}
Peter~E. Kloeden and Eckhard Platen.
\newblock \emph{{Numerical Solution of Stochastic Differential Equations}},
  volume~23 of \emph{Stochastic Modelling and Applied Probability}.
\newblock Springer-Verlag Berlin Heidelberg, 1992.

\bibitem[Kong et~al.(2020)Kong, Kim, and Bae]{HiFi-GAN}
Jungil Kong, Jaehyeon Kim, and Jaekyoung Bae.
\newblock {HiFi-GAN: Generative Adversarial Networks for Efficient and High
  Fidelity Speech Synthesis}.
\newblock In \emph{Advances in Neural Information Processing Systems 33: Annual
  Conference on Neural Information Processing Systems 2020, NeurIPS 2020,
  December 6-12, virtual}, 2020.

\bibitem[Kong \& Ping(2021)Kong and Ping]{FastSampling}
Zhifeng Kong and Wei Ping.
\newblock {On Fast Sampling of Diffusion Probabilistic Models}.
\newblock In \emph{ICML Workshop on Invertible Neural Networks, Normalizing
  Flows, and Explicit Likelihood Models}, 2021.

\bibitem[Kong et~al.(2021)Kong, Ping, Huang, Zhao, and Catanzaro]{DiffWave}
Zhifeng Kong, Wei Ping, Jiaji Huang, Kexin Zhao, and Bryan Catanzaro.
\newblock {DiffWave: A Versatile Diffusion Model for Audio Synthesis}.
\newblock In \emph{International Conference on Learning Representations}, 2021.

\bibitem[Lee et~al.(2021)Lee, Kim, Shin, Tan, Liu, Meng, Qin, Chen, Yoon, and
  Liu]{PriorGrad}
{Sang-gil} Lee, Heeseung Kim, Chaehun Shin, Xu~Tan, Chang Liu, Qi~Meng, Tao
  Qin, Wei Chen, Sungroh Yoon, and Tie-Yan Liu.
\newblock {PriorGrad: Improving Conditional Denoising Diffusion Models with
  Data-Driven Adaptive Prior}, 2021.

\bibitem[Lin et~al.(2021)Lin, Chien, Lin, Lee, and Lee]{FragmentVC}
Yist~Y. Lin, Chung{-}Ming Chien, Jheng{-}Hao Lin, Hung{-}yi Lee, and Lin{-}Shan
  Lee.
\newblock {FragmentVC: Any-To-Any Voice Conversion by End-To-End Extracting and
  Fusing Fine-Grained Voice Fragments with Attention}.
\newblock In \emph{ICASSP 2021 - 2021 IEEE International Conference on
  Acoustics, Speech and Signal Processing (ICASSP)}, pp.\  5939--5943, 2021.

\bibitem[Liptser \& Shiryaev(1978)Liptser and Shiryaev]{SDE-book}
Robert~S. Liptser and Albert~N. Shiryaev.
\newblock \emph{{Statistics of Random Processes}}, volume~5 of \emph{Stochastic
  Modelling and Applied Probability}.
\newblock Springer-Verlag, 1978.

\bibitem[Liu et~al.(2021{\natexlab{a}})Liu, Cao, Su, and Meng]{DiffSVC}
Songxiang Liu, Yuewen Cao, Dan Su, and Helen Meng.
\newblock {DiffSVC: A Diffusion Probabilistic Model for Singing Voice
  Conversion}, 2021{\natexlab{a}}.

\bibitem[Liu et~al.(2021{\natexlab{b}})Liu, Cao, Wang, Wu, Liu, and
  Meng]{PPG-VC}
Songxiang Liu, Yuewen Cao, Disong Wang, Xixin Wu, Xunying Liu, and Helen Meng.
\newblock {Any-to-Many Voice Conversion With Location-Relative
  Sequence-to-Sequence Modeling}.
\newblock \emph{IEEE/ACM Transactions on Audio, Speech, and Language
  Processing}, 29:\penalty0 1717--1728, 2021{\natexlab{b}}.

\bibitem[Luong \& Tran(2021)Luong and Tran]{DisentVAE}
Manh Luong and Viet~Anh Tran.
\newblock {Many-to-Many Voice Conversion Based Feature Disentanglement Using
  Variational Autoencoder}.
\newblock In \emph{Proc. Interspeech 2021}, pp.\  851--855, 2021.

\bibitem[McAuliffe et~al.(2017)McAuliffe, Socolof, Mihuc, Wagner, and
  Sonderegger]{MFA}
Michael McAuliffe, Michaela Socolof, Sarah Mihuc, Michael Wagner, and Morgan
  Sonderegger.
\newblock {Montreal Forced Aligner: Trainable Text-Speech Alignment Using
  Kaldi}.
\newblock In \emph{Proc. Interspeech 2017}, pp.\  498--502, 2017.

\bibitem[Nercessian(2020)]{AutoVC-PPG}
Shahan Nercessian.
\newblock {Improved Zero-Shot Voice Conversion Using Explicit Conditioning
  Signals}.
\newblock In \emph{Interspeech 2020, 21st Annual Conference of the
  International Speech Communication Association, Virtual Event, Shanghai,
  China, 25-29 October 2020}, pp.\  4711--4715. {ISCA}, 2020.

\bibitem[Nichol \& Dhariwal(2021)Nichol and Dhariwal]{IDDPM}
Alexander~Quinn Nichol and Prafulla Dhariwal.
\newblock {Improved Denoising Diffusion Probabilistic Models}.
\newblock In \emph{Proceedings of the 38th International Conference on Machine
  Learning}, volume 139 of \emph{Proceedings of Machine Learning Research},
  pp.\  8162--8171. PMLR, 18--24 Jul 2021.

\bibitem[Panayotov et~al.(2015)Panayotov, Chen, Povey, and
  Khudanpur]{LibriSpeech}
Vassil Panayotov, Guoguo Chen, Daniel Povey, and Sanjeev Khudanpur.
\newblock {Librispeech: An ASR Corpus Based on Public Domain Audio Books}.
\newblock In \emph{2015 IEEE International Conference on Acoustics, Speech and
  Signal Processing (ICASSP)}, pp.\  5206--5210, 2015.

\bibitem[Popov et~al.(2021)Popov, Vovk, Gogoryan, Sadekova, and
  Kudinov]{Grad-TTS}
Vadim Popov, Ivan Vovk, Vladimir Gogoryan, Tasnima Sadekova, and Mikhail
  Kudinov.
\newblock {Grad-TTS: A Diffusion Probabilistic Model for Text-to-Speech}.
\newblock In \emph{Proceedings of the 38th International Conference on Machine
  Learning, {ICML} 2021, 18-24 July 2021, Virtual Event}, volume 139 of
  \emph{Proceedings of Machine Learning Research}, pp.\  8599--8608. {PMLR},
  2021.

\bibitem[Qian et~al.(2019)Qian, Zhang, Chang, Yang, and
  Hasegawa-Johnson]{AutoVC}
Kaizhi Qian, Yang Zhang, Shiyu Chang, Xuesong Yang, and Mark Hasegawa-Johnson.
\newblock {AutoVC: Zero-Shot Voice Style Transfer with Only Autoencoder Loss}.
\newblock In \emph{Proceedings of the 36th International Conference on Machine
  Learning}, volume~97 of \emph{Proceedings of Machine Learning Research}, pp.\
   5210--5219. PMLR, 09--15 Jun 2019.

\bibitem[Qian et~al.(2020)Qian, Jin, Hasegawa-Johnson, and Mysore]{AutoVC-F0}
Kaizhi Qian, Zeyu Jin, Mark Hasegawa-Johnson, and Gautham~J. Mysore.
\newblock {F0-Consistent Many-To-Many Non-Parallel Voice Conversion Via
  Conditional Autoencoder}.
\newblock In \emph{ICASSP 2020 - 2020 IEEE International Conference on
  Acoustics, Speech and Signal Processing (ICASSP)}, pp.\  6284--6288, 2020.

\bibitem[Ronneberger et~al.(2015)Ronneberger, Fischer, and Brox]{UNet}
Olaf Ronneberger, Philipp Fischer, and Thomas Brox.
\newblock {U-Net: Convolutional Networks for Biomedical Image Segmentation}.
\newblock In \emph{Medical Image Computing and Computer-Assisted Intervention
  -- MICCAI 2015}, pp.\  234--241. Springer International Publishing, 2015.

\bibitem[Saito et~al.(2018)Saito, Ijima, Nishida, and Takamichi]{VAE-PPG}
Yuki Saito, Yusuke Ijima, Kyosuke Nishida, and Shinnosuke Takamichi.
\newblock {Non-Parallel Voice Conversion Using Variational Autoencoders
  Conditioned by Phonetic Posteriorgrams and D-Vectors}.
\newblock In \emph{2018 IEEE International Conference on Acoustics, Speech and
  Signal Processing (ICASSP)}, pp.\  5274--5278, 2018.

\bibitem[San-Roman et~al.(2021)San-Roman, Nachmani, and Wolf]{NoiseEstimation}
Robin San-Roman, Eliya Nachmani, and Lior Wolf.
\newblock {Noise Estimation for Generative Diffusion Models}, 2021.

\bibitem[Song et~al.(2021{\natexlab{a}})Song, Meng, and Ermon]{DDIM}
Jiaming Song, Chenlin Meng, and Stefano Ermon.
\newblock {Denoising Diffusion Implicit Models}.
\newblock In \emph{International Conference on Learning Representations},
  2021{\natexlab{a}}.

\bibitem[Song et~al.(2021{\natexlab{b}})Song, Durkan, Murray, and
  Ermon]{DPM-ML-training}
Yang Song, Conor Durkan, Iain Murray, and Stefano Ermon.
\newblock {Maximum Likelihood Training of Score-Based Diffusion Models},
  2021{\natexlab{b}}.

\bibitem[Song et~al.(2021{\natexlab{c}})Song, Sohl-Dickstein, Kingma, Kumar,
  Ermon, and Poole]{Song-main}
Yang Song, Jascha Sohl-Dickstein, Diederik~P Kingma, Abhishek Kumar, Stefano
  Ermon, and Ben Poole.
\newblock {Score-Based Generative Modeling through Stochastic Differential
  Equations}.
\newblock In \emph{International Conference on Learning Representations},
  2021{\natexlab{c}}.

\bibitem[Vincent(2011)]{Score-Matching2}
Pascal Vincent.
\newblock {A Connection Between Score Matching and Denoising Autoencoders}.
\newblock \emph{Neural Computation}, 23\penalty0 (7):\penalty0 1661--1674,
  2011.

\bibitem[Wang et~al.(2021)Wang, Deng, Yeung, Chen, Liu, and Meng]{VQMIVC}
Disong Wang, Liqun Deng, Yu~Ting Yeung, Xiao Chen, Xunying Liu, and Helen Meng.
\newblock {VQMIVC: Vector Quantization and Mutual Information-Based
  Unsupervised Speech Representation Disentanglement for One-Shot Voice
  Conversion}.
\newblock In \emph{Proc. Interspeech 2021}, pp.\  1344--1348, 2021.

\bibitem[Watson et~al.(2021)Watson, Ho, Norouzi, and Chan]{TrainableSampling}
Daniel Watson, Jonathan Ho, Mohammad Norouzi, and William Chan.
\newblock {Learning to Efficiently Sample from Diffusion Probabilistic Models},
  2021.

\bibitem[Wu et~al.(2020)Wu, Chen, and Lee]{VQVCp}
Da{-}Yi Wu, Yen{-}Hao Chen, and Hung{-}yi Lee.
\newblock {VQVC+: One-Shot Voice Conversion by Vector Quantization and U-Net
  Architecture}.
\newblock In Helen Meng, Bo~Xu, and Thomas~Fang Zheng (eds.), \emph{Interspeech
  2020, 21st Annual Conference of the International Speech Communication
  Association, Virtual Event, Shanghai, China, 25-29 October 2020}, pp.\
  4691--4695. {ISCA}, 2020.

\bibitem[Yamagishi et~al.(2019)Yamagishi, Veaux, and MacDonald]{VCTK}
Junichi Yamagishi, Christophe Veaux, and Kirsten MacDonald.
\newblock {CSTR VCTK Corpus: English Multi-speaker Corpus for CSTR Voice
  Cloning Toolkit (version 0.92)}, 2019.

\bibitem[Zen et~al.(2019)Zen, Clark, Weiss, Dang, Jia, Wu, Zhang, and
  Chen]{LibriTTS}
Heiga Zen, Rob Clark, Ron~J. Weiss, Viet Dang, Ye~Jia, Yonghui Wu, Yu~Zhang,
  and Zhifeng Chen.
\newblock {LibriTTS: A Corpus Derived from LibriSpeech for Text-to-Speech}.
\newblock In \emph{Interspeech}, 2019.

\end{thebibliography}
\bibliographystyle{iclr2022}

\newpage

\appendix
\section{Forward VP SDE solution}
\label{app:sde_solution}
Since function $\gamma_{0,t}^{-1}X_{t}$ is linear in $X_{t}$, taking its differential does not require second order derivative term in It\^o's formula:

\begin{equation}
\begin{split}
    d(\gamma_{0,t}^{-1}X_{t})&=d\left(e^{\frac{1}{2}\int_{0}^{t}{\beta_{u}du}}X_{t}\right)\\&=e^{\frac{1}{2}\int_{0}^{t}{\beta_{u}du}}\cdot \frac{1}{2}\beta_{t}X_{t}dt + e^{\frac{1}{2}\int_{0}^{t}{\beta_{u}du}}\cdot \left(-\frac{1}{2}\beta_{t}X_{t}dt + \sqrt{\beta_{t}}d\overrightarrow{W_{t}}\right) \\&= \sqrt{\beta_{t}}e^{\frac{1}{2}\int_{0}^{t}{\beta_{u}du}}d\overrightarrow{W_{t}}.
\end{split}
\end{equation}

Integrating this expression from $s$ to $t$ results in an It\^o's integral:

\begin{equation}
    e^{\frac{1}{2}\int_{0}^{t}{\beta_{u}du}}X_{t} - e^{\frac{1}{2}\int_{0}^{s}{\beta_{u}du}}X_{s} = \int_{s}^{t}{\sqrt{\beta_{\tau}}e^{\frac{1}{2}\int_{0}^{\tau}{\beta_{u}du}}d\overrightarrow{W_{\tau}}},
\end{equation}

or

\begin{equation}
    X_{t} = e^{-\frac{1}{2}\int_{s}^{t}{\beta_{u}du}}X_{s} + \int_{s}^{t}{\sqrt{\beta_{\tau}}e^{-\frac{1}{2}\int_{\tau}^{t}{\beta_{u}du}}d\overrightarrow{W_{\tau}}}.
\end{equation}

The integrand on the right-hand side is deterministic and belongs to $L_{2}[0,1]$ (for practical noise schedule choices), so its It\^o's integral is a normal random variable, a martingale (meaning it has zero mean) and satisfies It\^o's isometry which allows to calculate its variance:

\begin{equation}
    \Var(X_{t}|X_{s}) = \int_{s}^{t}{\beta_{\tau}}e^{-\int_{\tau}^{t}{\beta_{u}du}}\ide d\tau = \left(1 - e^{-\int_{s}^{t}{\beta_{u}du}}\right)\ide .
\end{equation}

Thus

\begin{equation}
    \Law(X_{t}|X_{s}) = \mathcal{N}\left(e^{-\frac{1}{2}\int_{s}^{t}{\beta_{u}du}}X_{s}, \left(1 - e^{-\int_{s}^{t}{\beta_{u}du}}\right)\ide\right) = \mathcal{N}(\gamma_{s,t}X_{s}, (1-\gamma_{s,t}^{2})\ide)
\end{equation}

\section{The optimal coefficients asymptotics}
\label{app:asymptotics}

First derive asymptotics for $\gamma$:
\begin{equation}
    \gamma_{t-h, t} = e^{-\frac{1}{2}\int_{t-h}^{t}{\beta_{u}du}} = 1 - \frac{1}{2}\beta_{t}h + \bar{o}(h),
\end{equation}
\begin{equation}
    \gamma_{0, t-h}^{2}=e^{-\int_{0}^{t-h}{\beta_{u}du}}=e^{-\int_{0}^{t}{\beta_{u}du}}e^{\int_{t-h}^{t}{\beta_{u}du}}=\gamma_{0,t}^{2}(1 + \beta_{t}h) + \bar{o}(h),
\end{equation}
\begin{equation}
    \gamma_{0, t-h}=e^{-\frac{1}{2}\int_{0}^{t-h}{\beta_{u}du}}=e^{-\frac{1}{2}\int_{0}^{t}{\beta_{u}du}}e^{\frac{1}{2}\int_{t-h}^{t}{\beta_{u}du}}=\gamma_{0,t}(1 + \frac{1}{2}\beta_{t}h) + \bar{o}(h),
\end{equation}
\begin{equation}
    \gamma_{t-h, t}^{2} = e^{-\int_{t-h}^{t}{\beta_{u}du}} = 1 - \beta_{t}h + \bar{o}(h).
\end{equation}

Then find asymptotics for $\mu$, $\nu$ and $\sigma^{2}$:
\begin{equation}
    \mu_{t-h,t}=\left(1 - \frac{1}{2}\beta_{t}h + \bar{o}(h)\right)\frac{1-\gamma_{0,t}^{2}-\gamma_{0,t}^{2}\beta_{t}h + \bar{o}(h)}{1-\gamma_{0,t}^{2}} = 1 - \frac{1}{2}\beta_{t}h-\frac{\gamma_{0,t}^{2}}{1-\gamma_{0,t}^{2}}\beta_{t}h + \bar{o}(h),
\end{equation}
\begin{equation}
    \nu_{t-h,t} = (\gamma_{0,t}(1 + \frac{1}{2}\beta_{t}h) + \bar{o}(h))\frac{\beta_{t}h + \bar{o}(h)}{1 - \gamma_{0,t}^{2}} = \frac{\gamma_{0,t}}{1 - \gamma_{0,t}^{2}}\beta_{t}h + \bar{o}(h),
\end{equation}
\begin{equation}
    \sigma^{2}_{t-h,t} = \frac{1}{1 - \gamma_{0,t}^{2}}(\beta_{t}h + \bar{o}(h))(1 - \gamma_{0,t}^{2}(1 + \beta_{t}h) + \bar{o}(h)) = \beta_{t}h + \bar{o}(h).
\end{equation}

Finally we get asymptotics for $\kappa^{*}$, $\omega^{*}$ and $\sigma^{*}$:

\begin{equation}
\begin{split}
    \kappa_{t,h}^{*}&=\frac{\nu_{t-h,t}(1 - \gamma_{0,t}^{2})}{\gamma_{0,t}\beta_{t}h} - 1 = \frac{\gamma_{0,t-h}(1 - \gamma_{t-h,t}^{2})}{\gamma_{0,t}\beta_{t}h} - 1 \\& =\frac{(\beta_{t}h + \bar{o}(h))((1 + \frac{1}{2}\beta_{t}h)\gamma_{0,t} + \bar{o}(h))}{\gamma_{0,t}\beta_{t}h} - 1 = \bar{o}(1),
\end{split}
\end{equation}

\begin{equation}
\begin{split}
    &\beta_{t} h \omega_{t,h}^{*} = \mu_{t-h,t} - 1 +  \frac{\nu_{t-h,t}}{\gamma_{0,t}} - \frac{1}{2}\beta_{t}h = 1 - \frac{1}{2}\beta_{t}h - \frac{\gamma_{0,t}^{2}}{1 - \gamma_{0,t}^{2}}\beta_{t}h - 1 - \frac{1}{2}\beta_{t}h + \frac{1}{\gamma_{0,t}}\times \\& \times\left(\frac{\gamma_{0,t}}{1 - \gamma_{0,t}^{2}}\beta_{t}h + \bar{o}(h)\right) + \bar{o}(h) = \beta_{t}h\left(-1 - \frac{\gamma_{0,t}^{2}}{1 - \gamma_{0,t}^{2}} + \frac{1}{1 - \gamma_{0,t}^{2}}\right) + \bar{o}(h) = \bar{o}(h),
\end{split}
\end{equation}

\begin{equation}
\begin{split}
    &(\sigma^{*}_{t,h})^{2} = \sigma^{2}_{t-h,t} + \nu_{t-h,t}^{2}\mathbb{E}_{X_{t}}\left[\Tr \left(\Var{(X_{0}|X_{t})}\right)\right]/n = \beta_{t}h + \bar{o}(h) \\
    &+ \frac{\gamma_{0,t}^{2}}{(1 - \gamma_{0,t}^{2})^{2}}\beta_{t}^{2}h^{2}\mathbb{E}_{X_{t}}\left[\Tr \left(\Var{(X_{0}|X_{t})}\right)\right]/n = \beta_{t}h(1 + \bar{o}(1)).
\end{split}
\end{equation}

\section{Proof of the Theorem \ref{th:main}}
\label{app:th_proof}

The key fact necessary to prove the Theorem \ref{th:main} is established in the following

\begin{lemma}
\label{lm:net_optimal}
    Let $p_{0|t}(\cdot|x)$ be pdf of conditional distribution $\Law{(X_{0}|X_{t}=x)}$. Then for any $t\in [0,1]$ and $x\in \mathbb{R}^{n}$
    \begin{equation}
    \label{eq:net_optimal}
        s_{\theta^{*}}(x,t) = -\frac{1}{1-\gamma_{0,t}^{2}}\left(x - \gamma_{0,t}\mathbb{E}_{p_{0|t}(\cdot|x)}X_{0}\right).
    \end{equation}
\end{lemma}

\begin{proof}[Proof of the Lemma \ref{lm:net_optimal}]
As mentioned in \citep{Song-main}, an expression alternative to (\ref{eq:objective}) can be derived for $\theta^{*}$ under mild assumptions on the data density \citep{Score-Matching1, Score-Matching2}:

\begin{equation}
    \theta^{*}=\argmin_{\theta}{\int_{0}^{1}{\lambda_{t}\mathbb{E}_{X_{0}\sim p_{0}(\cdot)}\mathbb{E}_{X_{t}\sim p_{t|0}(\cdot|X_{0})}\Vert s_{\theta}(X_{t}, t) - \nabla{\log{p_{t|0}(X_{t}|X_{0})}}\Vert_{2}^{2}}}dt,
\end{equation}
where $\Law{(X_{0})}$ is data distribution with pdf $p_{0}(\cdot)$ and $\Law{(X_{t}|X_{0}=x_{0})}$ has pdf $p_{t|0}(\cdot | x_{0})$. By Bayes formula we can rewrite this in terms of pdfs $p_{t}(\cdot)$ and $p_{0|t}(\cdot | x_{t})$ of distributions $\Law{(X_{t})}$ and $\Law{(X_{0}|X_{t}=x_{t})}$ correspondingly:

\begin{equation}
    \theta^{*}=\argmin_{\theta}{\int_{0}^{1}{\lambda_{t}\mathbb{E}_{X_{t}\sim p_{t}(\cdot)}\mathbb{E}_{X_{0}\sim p_{0|t}(\cdot|X_{t})}\Vert s_{\theta}(X_{t}, t) - \nabla{\log{p_{t|0}(X_{t}|X_{0})}}\Vert_{2}^{2}}}dt.
\end{equation}

For any $n$-dimensional random variable $\xi$ with finite second moment and deterministic vector $a$ we have 

\begin{equation}
\begin{split}
    \mathbb{E}\Vert\xi - a\Vert_{2}^{2} &= \mathbb{E}\Vert\xi - \mathbb{E}\xi + \mathbb{E}\xi - a\Vert_{2}^{2} = \mathbb{E}\Vert\xi - \mathbb{E}\xi\Vert_{2}^{2} + 2\langle\mathbb{E}[\xi - \mathbb{E}\xi],\mathbb{E}\xi - a\rangle \\&+ \mathbb{E}\Vert\mathbb{E}\xi - a\Vert_{2}^{2} = \mathbb{E}\Vert\xi - \mathbb{E}\xi\Vert_{2}^{2} + \Vert\mathbb{E}\xi - a\Vert_{2}^{2}.
\end{split}
\end{equation}

In our case $\xi=\nabla{\log{p_{t|0}(X_{t}|X_{0})}}$ and $a=s_{\theta}(X_{t}, t)$, so $\mathbb{E}\Vert\xi - \mathbb{E}\xi\Vert_{2}^{2}$ is independent of $\theta$. Thus

\begin{equation}
    \theta^{*}=\argmin_{\theta}{\int_{0}^{1}{\lambda_{t}\mathbb{E}_{X_{t}\sim p_{t}(\cdot)}\Vert s_{\theta}(X_{t}, t) - \mathbb{E}_{X_{0}\sim p_{0|t}(\cdot|X_{t})} \left[\nabla{\log{p_{t|0}(X_{t}|X_{0})}}\right]\Vert_{2}^{2}}}dt.
\end{equation}

Therefore, the optimal score estimation network $s_{\theta^{*}}$ can be expressed as

\begin{equation}
    s_{\theta^{*}}(x, t) = \mathbb{E}_{p_{0|t}(\cdot|x)} \left[\nabla{\log{p_{t|0}(x|X_{0})}}\right]
\end{equation}

for all $t\in [0,1]$ and $x\in \supp{\{p_{t}\}}=\mathbb{R}^{n}$.

As proven in Appendix \ref{app:sde_solution}, $\Law{(X_{t}|X_{0})}$ is Gaussian with mean vector $\gamma_{0,t}X_{0}$ and covariance matrix $(1-\gamma_{0,t}^{2})\ide$, so finally we obtain

\begin{equation}
    s_{\theta^{*}}(x,t) =
    \mathbb{E}_{p_{0|t}(\cdot|x)}\left[-\frac{1}{1-\gamma_{0,t}^{2}}\left(x - \gamma_{0,t}X_{0}\right)\right]= -\frac{1}{1-\gamma_{0,t}^{2}}\left(x - \gamma_{0,t}\mathbb{E}_{p_{0|t}(\cdot|x)}X_{0}\right).
\end{equation}

\end{proof}

Now let us prove the Theorem \ref{th:main}.

\begin{proof}[Proof of the Theorem \ref{th:main}]
    The sampling scheme (\ref{eq:sde_solver}) consists in adding Gaussian noise to a linear combination of $\hat{X}_{t}$ and $s_{\theta^{*}}(\hat{X}_{t}, t)$. Combining (\ref{eq:sde_solver}) and the Lemma \ref{lm:net_optimal} we get

    \begin{equation}
    \begin{split}
        &\hat{X}_{t-h} = \hat{\sigma}_{t,h}\xi_{t} + \hat{X}_{t} + \beta_{t}h\left(\left(\frac{1}{2}+\hat{\omega}_{t,h}\right)\hat{X}_{t} + (1+\hat{\kappa}_{t,h})s_{\theta^{*}}(\hat{X}_{t},t)\right) = \hat{\sigma}_{t,h}\xi_{t} \\ &+ \left(1 + \beta_{t}h\left(\frac{1}{2} + \hat{\omega}_{t,h}\right)\right)\hat{X}_{t} + \beta_{t}h(1+\hat{\kappa}_{t,h})\left(-\frac{1}{1-\gamma_{0,t}^{2}}\left(\hat{X}_{t} - \gamma_{0,t}\mathbb{E}_{p_{0|t}(\cdot|\hat{X}_{t})}X_{0}\right)\right) \\& =\hat{\sigma}_{t,h}\xi_{t} + \left(1 + \beta_{t}h\left(\frac{1}{2} + \hat{\omega}_{t,h}-\frac{1+\hat{\kappa}_{t,h}}{1-\gamma_{0,t}^{2}}\right)\right)\hat{X}_{t} + \frac{\gamma_{0,t}\beta_{t}h(1+\hat{\kappa}_{t,h})}{1-\gamma_{0,t}^{2}}\mathbb{E}_{p_{0|t}(\cdot|\hat{X}_{t})}X_{0},
    \end{split}
    \end{equation}
    where $\xi_{t}$ are i.i.d. random variables from standard normal distribution $\mathcal{N}(0,\ide)$ for $t=1,1-h,..,h$. Thus, the distribution $\hat{X}_{t-h}|\hat{X}_{t}$ is also Gaussian:

    \begin{equation}
    \label{eq:law_net}
        \Law{(\hat{X}_{t-h} | \hat{X}_{t})} = \mathcal{N}\left(\hat{\mu}_{t,h}(\hat{\kappa}_{t,h},\hat{\omega}_{t,h})\hat{X}_{t} + \hat{\nu}_{t,h}(\hat{\kappa}_{t,h})\mathbb{E}_{p_{0|t}(\cdot|\hat{X}_{t})}X_{0},\hat{\sigma}_{t,h}^{2}\ide\right),
    \end{equation}

    \begin{equation}
        \hat{\mu}_{t,h}(\hat{\kappa}_{t,h},\hat{\omega}_{t,h}) = 1 + \beta_{t}h\left(\frac{1}{2} + \hat{\omega}_{t,h}-\frac{1+\hat{\kappa}_{t,h}}{1-\gamma_{0,t}^{2}}\right),
    \end{equation}

    \begin{equation}
        \hat{\nu}_{t,h}(\hat{\kappa}_{t,h}) = \frac{\gamma_{0,t}\beta_{t}h(1+\hat{\kappa}_{t,h})}{1-\gamma_{0,t}^{2}},
    \end{equation}

    which leads to the following formula for the transition densities of the reverse diffusion: 

    \begin{equation}
    \label{eq:p_net}
        \hat{p}_{t-h|t}(x_{t-h}|x_{t})=\frac{1}{\sqrt{2\pi}\hat{\sigma}_{t,h}^{n}}\exp\left(-\frac{\Vert x_{t-h} - \hat{\mu}_{t,h}x_{t} - \hat{\nu}_{t,h}\mathbb{E}_{p_{0|t}(\cdot|x_{t})}X_{0}\Vert_{2}^{2}}{2\hat{\sigma}^{2}_{t,h}}\right).
    \end{equation}
    
    Moreover, comparing $\hat{\mu}_{t,h}$ and $\hat{\nu}_{t,h}$ with $\mu_{t-h,t}$ and $\nu_{t-h,t}$ defined in (\ref{eq:notation_gamma}) we deduce that

    \begin{equation}
        \hat{\nu}_{t,h} = \nu_{t-h,t} \Leftrightarrow   \frac{\gamma_{0,t}\beta_{t}h(1+\hat{\kappa}_{t,h})}{1-\gamma_{0,t}^{2}} = \nu_{t-h,t} \Leftrightarrow \hat{\kappa}_{t,h}=\kappa^{*}_{t,h}.
    \end{equation}

    If we also want $\hat{\mu}_{t,h}=\mu_{t-h,t}$ to be satisfied, then we should have

    \begin{equation}
        1 + \beta_{t}h\left(\frac{1}{2} + \hat{\omega}_{t,h}-\frac{1+\kappa^{*}_{t,h}}{1-\gamma_{0,t}^{2}}\right) = \mu_{t-h,t} \Leftrightarrow \left(\frac{\mu_{t-h,t} - 1}{\beta_{t}h}-\omega^{*}_{t,h}+\hat{\omega}_{t,h}\right)\beta_{t}h + 1 = \mu_{t-h,t},
    \end{equation}

    i.e. $\hat{\nu}_{t,h}=\nu_{t-h,t}$ and $\hat{\mu}_{t,h}=\mu_{t-h,t}$ iff $\hat{\kappa}_{t,h}=\kappa^{*}_{t,h}$ and $\hat{\omega}_{t,h}=\omega^{*}_{t,h}$ for the parameters $\kappa^{*}_{t,h}$ and $\omega^{*}_{t,h}$ defined in (\ref{eq:notation_correction}).

    As for the corresponding densities of the forward process $X$, they are Gaussian when conditioned on the initial data point $X_{0}$:

    \begin{equation}
    \label{eq:law_true}
        \Law{(X_{t-h} | X_{t}, X_{0})} = \mathcal{N}(\mu_{t-h,t}X_{t} + \nu_{t-h,t}X_{0},\sigma_{t-h,t}^{2}\ide),
    \end{equation}

    where coefficients $\mu_{t-h,t}$, $\nu_{t-h,t}$ and $\sigma_{t-h,t}$ are defined in (\ref{eq:notation_gamma}). This formula for $\Law{(X_{t-h} | X_{t}, X_{0})}$ follows from the general fact about Gaussian distributions appearing in many recent works on diffusion probabilistic modeling \citep{VDPM}: if $Z_{t}|Z_{s}\sim \mathcal{N}(\alpha_{t|s}Z_{s}, \sigma_{t|s}^{2}\ide)$ and $Z_{t}|Z_{0}\sim \mathcal{N}(\alpha_{t|0}Z_{0}, \sigma_{t|0}^{2}\ide)$ for $0 < s < t$, then
    \begin{equation}
    \label{eq:gauss}
        \Law{(Z_{s}|Z_{t}, Z_{0})} = \mathcal{N}\left(\frac{\sigma_{s|0}^{2}}{\sigma_{t|0}^{2}}\alpha_{t|s}Z_{t} + \frac{\sigma_{t|s}^{2}}{\sigma_{t|0}^{2}}\alpha_{s|0}Z_{0}, \frac{\sigma_{s|0}^{2}\sigma_{t|s}^{2}}{\sigma_{t|0}^{2}}\ide\right).
    \end{equation}

    This fact is a result of applying Bayes formula to normal distributions. In our case $\alpha_{t|s}=\gamma_{s,t}$ and $\sigma_{t|s}^{2} = 1 - \gamma_{s,t}^{2}$.

    To get an expression for the densities $p_{t-h|t}(x_{t-h}|x_{t})$ similar to (\ref{eq:p_net}), we need to integrate out the dependency on data $X_{0}$ from Gaussian distribution $\Law{(X_{t-h}|X_{t}, X_{0})}$:

    \begin{equation}
    \label{eq:p_exc}
    \begin{split}
        p_{t-h|t}(x_{t-h}|x_{t}) = & \int{p_{t-h,0|t}(x_{t-h},x_{0}|x_{t})dx_{0}}=\int{p_{t-h|t,0}(x_{t-h}|x_{t},x_{0})p_{0|t}(x_{0}|x_{t})dx_{0}} \\&= \mathbb{E}_{X_{0}\sim p_{0|t}(\cdot|x_{t})}[p_{t-h|t,0}(x_{t-h}|x_{t},X_{0})],
    \end{split}
    \end{equation}

    which implies the following formula:

    \begin{equation}
    \label{eq:p_true}
        p_{t-h|t}(x_{t-h}|x_{t})=\frac{1}{\sqrt{2\pi}\sigma_{t-h,t}^{n}}\mathbb{E}_{p_{0|t}(\cdot|x_{t})}\left[\exp\left(-\frac{\Vert x_{t-h} - \mu_{t-h,t}x_{t} - \nu_{t-h,t}X_{0}\Vert_{2}^{2}}{2\sigma^{2}_{t-h,t}}\right)\right].
    \end{equation}

    Note that in contrast with the transition densities (\ref{eq:p_net}) of the reverse process $\hat{X}$, the corresponding densities (\ref{eq:p_true}) of the forward process $X$ are not normal in general.

    Our goal is to find parameters $\hat{\kappa}$, $\hat{\omega}$ and $\hat{\sigma}$ that maximize log-likelihood of sample paths $X$ under probability measure with transition densities $\hat{p}$. Put $t_{k}=kh$ for $k=0,1,..,N$ and write down this log-likelihood:

    \begin{equation}
    \begin{split}
        &\int{p(x_{1},x_{1-h},..,x_{0})\left(\sum_{k=0}^{N-1}{\log{\hat{p}_{t_{k}|t_{k+1}}}(x_{t_{k}}|x_{t_{k+1}})} + \log{\hat{p}_{1}(x_{1})}\right)dx_{1}dx_{1-h}..dx_{0}} \\
        &=\sum_{k=0}^{N-1}\int{p(x_{t_{k}},x_{t_{k+1}})\log{\hat{p}_{t_{k}|t_{k+1}}(x_{t_{k}}|x_{t_{k+1}})}dx_{t_{k+1}}dx_{t_{k}}} + \int{p(x_{1})\log{\hat{p}_{1}}(x_{1})dx_{1}}.
    \end{split}
    \end{equation}

    The last term does not depend on $\hat{\kappa}$, $\hat{\omega}$ and $\hat{\sigma}$, so we can ignore it. Let $R_{k}$ be the $k$-th term in the sum above. Since we are free to have different coefficients $\hat{\kappa}_{t,h}$, $\hat{\omega}_{t,h}$ and $\hat{\sigma}_{t,h}$ for different steps, we can maximize each $R_{k}$ separately. Terms $R_{k}$ can be expressed as

    \begin{equation}
    \begin{split}
        R_{k} &= \int{p(x_{t_{k}},x_{t_{k+1}})\log{\hat{p}_{t_{k}|t_{k+1}}(x_{t_{k}}|x_{t_{k+1}})}dx_{t_{k+1}}dx_{t_{k}}} \\
        &=\int{p(x_{t_{k+1}})p_{t_{k}|t_{k+1}}(x_{t_{k}}|x_{t_{k+1}})\log{\hat{p}_{t_{k}|t_{k+1}}(x_{t_{k}}|x_{t_{k+1}})}dx_{t_{k+1}}dx_{t_{k}}} \\
        &=\mathbb{E}_{X_{t_{k+1}}}\left[\int{p_{t_{k}|t_{k+1}}(x_{t_{k}}|X_{t_{k+1}})\log{\hat{p}_{t_{k}|t_{k+1}}(x_{t_{k}}|X_{t_{k+1}})}dx_{t_{k}}}\right].
    \end{split}
    \end{equation}

    From now on we will skip subscripts of $\mu$, $\nu$, $\sigma$, $\hat{\mu}$, $\hat{\nu}$, $\hat{\sigma}$, $\hat{\kappa}$, $\hat{\omega}$, $\kappa^{*}$ and $\omega^{*}$ for brevity. Denote

    \begin{equation}
    \label{eq:q}
        Q(x_{t_{k}},X_{t_{k+1}},X_{0})=\frac{1}{\sqrt{2\pi}\sigma^{n}}\exp\left(-\frac{\Vert x_{t_{k}} - \mu X_{t_{k+1}} - \nu X_{0}\Vert_{2}^{2}}{2\sigma^{2}}\right)\log{\hat{p}_{t_{k}|t_{k+1}}(x_{t_{k}}|X_{t_{k+1}})}.
    \end{equation}

    Using the formula (\ref{eq:p_exc}) for the densities of $X$ together with the explicit expression for the Gaussian density $p_{t_{k}|t_{k+1},0}(x_{t_{k}}|X_{t_{k+1}},X_{0})$ and applying Fubini's theorem to change the order of integration, we rewrite $R_{k}$ as 

    \begin{equation}
    \begin{split}
        R_{k}&=\mathbb{E}_{X_{t_{k+1}}}\left[\int{p_{t_{k}|t_{k+1}}(x_{t_{k}}|X_{t_{k+1}})}\log{\hat{p}_{t_{k}|t_{k+1}}(x_{t_{k}}|X_{t_{k+1}})}dx_{t_{k}}\right]\\
        &=\mathbb{E}_{X_{t_{k+1}}}\left[\int{\mathbb{E}_{X_{0}\sim p_{0|t_{k+1}}(\cdot|X_{t_{k+1}})}\left[p_{t_{k}|t_{k+1},0}(x_{t_{k}}|X_{t_{k+1}},X_{0})\log{\hat{p}_{t_{k}|t_{k+1}}(x_{t_{k}}|X_{t_{k+1}})}\right]dx_{t_{k}}}\right] \\
        &=\mathbb{E}_{X_{t_{k+1}}}\left[\int{\mathbb{E}_{X_{0}\sim p_{0|t_{k+1}}(\cdot|X_{t_{k+1}})}[Q(x_{t_{k}}, X_{t_{k+1}}, X_{0})]dx_{t_{k}}}\right]\\
        &=\mathbb{E}_{X_{t_{k+1}}}\mathbb{E}_{X_{0}\sim p_{0|t_{k+1}}(\cdot|X_{t_{k+1}})}\left[\int{Q(x_{t_{k}}, X_{t_{k+1}}, X_{0})dx_{t_{k}}}\right].
    \end{split}
    \end{equation}
    
    The formula (\ref{eq:q}) implies that the integral of $Q(x_{t_{k}}, X_{t_{k+1}},X_{0})$ with respect to $x_{t_{k}}$ can be seen as expectation of $\log{\hat{p}_{t_{k}|t_{k+1}}(\xi|X_{t_{k+1}})}$ with respect to normal random variable $\xi$ with mean $\mu X_{t_{k+1}} + \nu X_{0}$ and covariance matrix $\sigma^{2}\ide$. Plugging in the expression (\ref{eq:p_net}) into (\ref{eq:q}), we can calculate this integral:

    \begin{equation}
    \begin{split}
        &\mathbb{E}_{\xi}\left[-\log{\sqrt{2\pi}} -n\log{\hat{\sigma}} - \frac{\Vert\xi - \hat{\mu}X_{t_{k+1}}-\hat{\nu}\mathbb{E}_{X_{0}'\sim p_{0|t_{k+1}}(\cdot|X_{t_{k+1}})}X_{0}'\Vert_{2}^{2}}{2\hat{\sigma}^{2}}\right] \\
        &=-\log{\sqrt{2\pi}} - n\log{\hat{\sigma}} - \frac{\mathbb{E}_{\xi}\Vert \xi - \hat{\mu}X_{t_{k+1}}-\hat{\nu}\mathbb{E}_{p_{0|t_{k+1}}(\cdot|X_{t_{k+1}})}X_{0}'\Vert_{2}^{2}}{2\hat{\sigma}^{2}}.
    \end{split}
    \end{equation}

    Thus, terms $R_{k}$ we want to maximize equal

    \begin{equation}
        R_{k} = -\log{\sqrt{2\pi}} - n\log{\hat{\sigma}} - \mathbb{E}_{X_{t_{k+1}}}\mathbb{E}_{X_{0}\sim p_{0|t_{k+1}}(\cdot|X_{t_{k+1}})}\frac{\mathbb{E}_{\xi}\Vert \xi - \hat{\mu}X_{t_{k+1}}-\hat{\nu}\mathbb{E}_{p_{0|t_{k+1}}(\cdot|X_{t_{k+1}})}X_{0}'\Vert_{2}^{2}}{2\hat{\sigma}^{2}}
    \end{equation}

    Maximizing $R_{k}$ with respect to ($\hat{\kappa}, \hat{\omega}, \hat{\sigma}$) is equivalent to minimizing $\mathbb{E}_{X_{t_{k+1}}}S_{k}$ where $S_{k}$ is given by

    \begin{equation}
    \label{eq:s_term}
        S_{k} = n\log{\hat{\sigma}} + \frac{1}{2\hat{\sigma}^{2}}\mathbb{E}_{X_{0}\sim p_{0|t_{k+1}}(\cdot|X_{t_{k+1}})}\mathbb{E}_{\xi}\Vert \xi - \hat{\mu}X_{t_{k+1}}-\hat{\nu}\mathbb{E}_{p_{0|t_{k+1}}(\cdot|X_{t_{k+1}})}X_{0}'\Vert_{2}^{2},
    \end{equation}

    where the expectation with respect to $\xi\sim \mathcal{N}(\mu X_{t_{k+1}} + \nu X_{0}, \sigma^{2}\ide)$ can be calculated using the fact that for every vector $\hat{a}$ we can express $\mathbb{E}_{\xi}\Vert \xi - \hat{a} \Vert_{2}^{2}$ as

    \begin{equation}
        \mathbb{E}\Vert\xi - \mathbb{E}\xi + \mathbb{E}\xi - \hat{a}\Vert_{2}^{2} = \mathbb{E}\Vert\xi - \mathbb{E}\xi \Vert_{2}^{2} + 2\langle\mathbb{E}[\xi - \mathbb{E}\xi], \mathbb{E}\xi - \hat{a}\rangle + \mathbb{E}\Vert \mathbb{E}\xi - \hat{a}\Vert_{2}^{2} = n\sigma^{2} + \Vert \mathbb{E}\xi - \hat{a}\Vert_{2}^{2}.
    \end{equation}

    So, the outer expectation with respect to $\Law(X_{0}|X_{t_{k+1}})$ in (\ref{eq:s_term}) can be simplified:

    \begin{equation}
    \begin{split}
        &\mathbb{E}_{X_{0}\sim p_{0|t_{k+1}}(\cdot|X_{t_{k+1}})}\left[n\sigma^{2} + \Vert(\mu - \hat{\mu})X_{t_{k+1}} + \nu X_{0} - \hat{\nu}\mathbb{E}_{X_{0}'\sim p_{0|t_{k+1}}(\cdot|X_{t_{k+1}})}X_{0}'\Vert^{2}_{2}\right] \\
        &=n\sigma^{2} + \mathbb{E}_{X_{0}}\Vert((\mu - \hat{\mu})X_{t_{k+1}} + \nu X_{0}-\hat{\nu}\mathbb{E}_{X_{0}'}X_{0}'\Vert_{2}^{2}=n\sigma^{2} + (\mu - \hat{\mu})^{2}\Vert X_{t_{k+1}}\Vert_{2}^{2} \\
        &+2\langle(\mu - \hat{\mu})X_{t_{k+1}},(\nu - \hat{\nu})\mathbb{E}_{X_{0}}X_{0}\rangle + \mathbb{E}_{X_{0}}{\Vert\nu X_{0}} - \hat{\nu}\mathbb{E}_{X_{0}'}X_{0}'\Vert_{2}^{2} = (\mu - \hat{\mu})^{2}\Vert X_{t_{k+1}}\Vert_{2}^{2} \\
        & + 2\langle(\mu - \hat{\mu})X_{t_{k+1}},(\nu - \hat{\nu})\mathbb{E}_{X_{0}}X_{0}\rangle + \nu^{2}\mathbb{E}_{X_{0}}\Vert X_{0}\Vert_{2}^{2} + \hat{\nu}^{2}\Vert\mathbb{E}_{X_{0}}X_{0}\Vert_{2}^{2} + n\sigma^{2} \\
        & -2\nu\hat{\nu}\langle\mathbb{E}_{X_{0}}X_{0},\mathbb{E}_{X'_{0}}X'_{0}\rangle = \Vert(\mu - \hat{\mu})X_{t_{k+1}} + (\nu - \hat{\nu})\mathbb{E}_{X_{0}}X_{0})\Vert_{2}^{2} + \nu^{2}\mathbb{E}_{X_{0}}\Vert X_{0}\Vert_{2}^{2} \\
        & -\nu^{2}\Vert\mathbb{E}_{X_{0}}X_{0}\Vert_{2}^{2} + n\sigma^{2},
    \end{split}
    \end{equation}

    where all the expectations in the formula above are taken with respect to the conditional data distribution $\Law(X_{0}|X_{t_{k+1}})$. So, the resulting expression for the terms $S_{k}$ whose expectation with respect to $\Law{(X_{t_{k+1}})}$ we want to minimize is

    \begin{equation}
    \label{eq:neg-log-lok-term}
    \begin{split}
        S_{k} = n\log&{\hat{\sigma}} + \frac{1}{2\hat{\sigma}^{2}}\Bigl(n\sigma^{2}+\Vert(\mu - \hat{\mu})X_{t_{k+1}} + (\nu - \hat{\nu})\mathbb{E}{[X_{0}|X_{t_{k+1}}]}\Vert_{2}^{2} \\& + \nu^{2}\left(\mathbb{E}\left[\Vert X_{0}\Vert_{2}^{2}|X_{t_{k+1}}\right]-\Vert\mathbb{E}[X_{0}|X_{t_{k+1}}]\Vert_{2}^{2}\right)\Bigr).
    \end{split}
    \end{equation}

    Now it is clear that $\kappa^{*}_{t_{k+1},h}$ and $\omega^{*}_{t_{k+1},h}$ are optimal because $\hat{\mu}_{t_{k+1},h}(\kappa^{*}_{t_{k+1},h}, \omega^{*}_{t_{k+1},h}) = \mu_{t_{k},t_{k+1}}$ and $\hat{\nu}_{t_{k+1},h}(\kappa^{*}_{t_{k+1},h}) = \nu_{t_{k},t_{k+1}}$. For this choice of parameters we have

    \begin{equation}
        \mathbb{E}_{X_{t_{k+1}}}S_{k} = n\log{\hat{\sigma}} + \frac{1}{2\hat{\sigma}^{2}}\left(n\sigma^{2}+\nu^{2}\mathbb{E}_{X_{t_{k+1}}}\left[\mathbb{E}\left[\Vert X_{0}\Vert_{2}^{2}|X_{t_{k+1}}\right]-\Vert\mathbb{E}[X_{0}|X_{t_{k+1}}]\Vert_{2}^{2}\right]\right).
    \end{equation}

    Note that $\mathbb{E}\left[\Vert X_{0}\Vert_{2}^{2}|X_{t_{k+1}}\right]-\Vert\mathbb{E}[X_{0}|X_{t_{k+1}}]\Vert_{2}^{2} = \Tr{(\Var{(X_{0}|X_{t_{k+1}})})}$ is the overall variance of $\Law{(X_{0}|X_{t_{k+1}})}$ along all $n$ dimensions. Differentiating $\mathbb{E}_{X_{t_{k+1}}}{S_{k}}$ with respect to $\hat{\sigma}$ shows that the optimal $\sigma^{*}_{t_{k+1},h}$ should satisfy

    \begin{equation}
        \frac{n}{\sigma^{*}_{t_{k+1},h}} - \frac{1}{(\sigma^{*}_{t_{k+1},h})^{3}}\left(n\sigma_{t_{k},t_{k+1}}^{2}+\nu_{t_{k},t_{k+1}}^{2}\mathbb{E}_{X_{t_{k+1}}}\left[\Tr{(\Var{(X_{0}|X_{t_{k+1}})})}\right]\right)=0,
    \end{equation}

    which is indeed satisfied by the parameters $\sigma^{*}_{t,h}$ defined in (\ref{eq:notation_correction}). Thus, the statement (i) is proven.

    When it comes to proving that $\hat{X}$ is exact, we have to show that $\Law{(\hat{X}_{t_{k}})}=\Law{(X_{t_{k}})}$ for every $k=0,1,..,N$. By the assumption that $\Law{(\hat{X}_{1})} = \Law{(X_{1})}$ it is sufficient to prove that $\hat{p}_{t_{k}|t_{k+1}}(x_{t_{k}}|x_{t_{k+1}}) \equiv p_{t_{k}|t_{k+1}}(x_{t_{k}}|x_{t_{k+1}})$ since the exactness will follow from this fact by mathematical induction. If $X_{0}$ is a constant random variable, $\Law(X_{0}|X_{t})=\Law{(X_{0})}$ also corresponds to the same constant, so $\Var{(X_{0}|X_{t})}=0$ meaning that $\sigma^{*}_{t,h}=\sigma_{t-h,t}$, and the formulae (\ref{eq:p_net}) and (\ref{eq:p_true}) imply the desired result. 

    Let us now consider the second case when $X_{0}\sim \mathcal{N}(\bar{\mu},\delta^{2}\ide)$. It is a matter of simple but lengthy computations to prove another property of Gaussian distributions similar to (\ref{eq:gauss}): if $Z_{0}\sim \mathcal{N}(\bar{\mu},\delta^{2}\ide)$ and $Z_{t}|Z_{0}\sim \mathcal{N}(a_{t}Z_{0},b^{2}_{t}\ide)$, then $Z_{0}|Z_{t}\sim \mathcal{N}(\frac{b_{t}^{2}}{b^{2}_{t}+\delta^{2}a^{2}_{t}}\bar{\mu} + \frac{\delta^{2}a_{t}}{b^{2}_{t}+\delta^{2}a^{2}_{t}}Z_{t},\frac{\delta^{2}b^{2}_{t}}{b_{t}^{2}+\delta^{2}a^{2}_{t}}\ide)$ and $Z_{t}\sim \mathcal{N}(\bar{\mu}a_{t},{(b^{2}_{t}+\delta^{2}a^{2}_{t}})\ide)$. In our case $a_{t}=\gamma_{0,t}$ and $b^{2}_{t}=1-\gamma_{0,t}^{2}$, therefore

    \begin{equation}
    \label{eq:gaussian_revert}
        \Law{(X_{0}|X_{t})}=\mathcal{N}\left(\frac{1-\gamma^{2}_{0,t}}{1-\gamma_{0,t}^{2}+\delta^{2}\gamma^{2}_{0,t}}\bar{\mu}+\frac{\delta^{2}\gamma_{0,t}}{1-\gamma_{0,t}^{2}+\delta^{2}\gamma^{2}_{0,t}}X_{t},\frac{\delta^{2}(1-\gamma_{0,t}^{2})}{1-\gamma_{0,t}^{2}+\delta^{2}\gamma^{2}_{0,t}}\ide\right).
    \end{equation}

    So, $\Var{(X_{0}|X_{t})}$ does not depend on $X_{t}$ and 

    \begin{equation}
        (\sigma^{*}_{t,h})^{2} = \sigma_{t-h,t}^{2} + \frac{\nu_{t-h,t}^{2}}{n}\mathbb{E}_{X_{t}}\left[\Tr{(\Var{(X_{0}|X_{t})})}\right] = \sigma_{t-h,t}^{2} + \nu_{t-h,t}^{2}\frac{\delta^{2}(1-\gamma_{0,t}^{2})}{1-\gamma_{0,t}^{2}+\delta^{2}\gamma^{2}_{0,t}}.
    \end{equation}

    Since $\Law{(X_{t}|X_{t-h})}$, $\Law{(X_{t-h})}$ and $\Law{(X_{t})}$ are Gaussian, Bayes formula implies that $\Law{(X_{t-h}|X_{t})}$ is Gaussian as well with the following mean and covariance matrix:

    \begin{equation}
        \mathbb{E}[X_{t-h}|X_{t}] = \frac{\gamma_{0,t-h}(1-\gamma_{t-h,t}^{2})}{1-\gamma^{2}_{0,t}+\delta^{2}\gamma_{0,t}^{2}}\bar{\mu} + \frac{\gamma_{t-h,t}(1-\gamma^{2}_{0,t-h}+\delta^{2}\gamma_{0,t-h}^{2})}{1-\gamma^{2}_{0,t}+\delta^{2}\gamma_{0,t}^{2}}X_{t},
    \end{equation}

    \begin{equation}
        \Var{(X_{t-h}|X_{t})}=\frac{(1-\gamma_{t-h,t}^{2})(1-\gamma_{0,t-h}^{2}+\delta^{2}\gamma_{0,t-h}^{2})}{1-\gamma_{0,t}^{2}+\delta^{2}\gamma_{0,t}^{2}}\ide.
    \end{equation}

    The distribution $\Law{(\hat{X}_{t-h}|\hat{X}_{t})}$ is also Gaussian by the formula (\ref{eq:p_net}), so to conclude the proof we just need to show that $\mathbb{E}[\hat{X}_{t-h}|\hat{X}_{t}=x]=\mathbb{E}[X_{t-h}|X_{t}=x]$ and $\Var{(\hat{X}_{t-h}|\hat{X}_{t}=x)}=\Var{(X_{t-h}|X_{t}=x)}$ for every $x\in \mathbb{R}^{n}$ for the optimal parameters (\ref{eq:notation_correction}). Recall that for $\kappa^{*}_{t,h}$ and $\omega^{*}_{t,h}$ we have $\hat{\mu}_{t,h}(\kappa^{*}_{t,h}, \omega^{*}_{t,h}) = \mu_{t-h,t}$ and $\hat{\nu}_{t,h}(\kappa^{*}_{t,h}) = \nu_{t-h,t}$. Utilizing the formulae (\ref{eq:notation_gamma}), (\ref{eq:p_net}), (\ref{eq:gaussian_revert}) and the fact that $\gamma_{0,t-h}\cdot\gamma_{t-h,t}=\gamma_{0,t}$ (following from the definition of $\gamma$ in (\ref{eq:xt_distribution})) we conclude that

    \begin{equation}
    \begin{split}
        &\mathbb{E}[\hat{X}_{t-h}|\hat{X}_{t}=x]=\hat{\mu}_{t,h}(\kappa^{*}_{t,h}, \omega^{*}_{t,h})x + \hat{\nu}_{t,h}(\kappa^{*}_{t,h})\mathbb{E}_{p_{0|t}(\cdot|x)}X_{0} \\
        & = \gamma_{t-h,t}\frac{1-\gamma_{0,t-h}^{2}}{1-\gamma_{0,t}^{2}}x + \gamma_{0,t-h}\frac{1-\gamma_{t-h,t}^{2}}{1-\gamma_{0,t}^{2}}\left[\frac{1-\gamma_{0,t}^{2}}{1-\gamma_{0,t}^{2}+\delta^{2}\gamma_{0,t}^{2}}\bar{\mu} + \frac{\delta^{2}\gamma_{0,t}}{1-\gamma_{0,t}^{2}+\delta^{2}\gamma_{0,t}^{2}}x\right] \\
        &= \gamma_{0,t-h}\frac{1-\gamma_{t-h,t}^{2}}{1-\gamma_{0,t}^{2} + \delta^{2}\gamma_{0,t}^{2}}\bar{\mu} + \gamma_{t-h,t}\frac{(1-\gamma_{0,t-h}^{2})(1-\gamma_{0,t}^{2})+\delta^{2}\gamma_{0,t}^{2}(1-\gamma^{2}_{0,t-h})}{(1-\gamma_{0,t}^{2})(1-\gamma_{0,t}^{2}+\delta^{2}\gamma_{0,t}^{2})}x \\
        &+\gamma_{t-h,t}\frac{\delta^{2}\gamma_{0,t-h}^{2}(1-\gamma^{2}_{t-h,t})}{(1-\gamma_{0,t}^{2})(1-\gamma_{0,t}^{2}+\delta^{2}\gamma_{0,t}^{2})}x = \gamma_{0,t-h}\frac{1-\gamma_{t-h,t}^{2}}{1-\gamma_{0,t}^{2} + \delta^{2}\gamma_{0,t}^{2}}\bar{\mu} \\
        &+\gamma_{t-h,t}\frac{(1-\gamma_{0,t-h}^{2})(1-\gamma_{0,t}^{2})+\delta^{2}\gamma_{0,t-h}^{2}(1-\gamma_{0,t}^{2})}{(1-\gamma_{0,t}^{2})(1-\gamma_{0,t}^{2}+\delta^{2}\gamma_{0,t}^{2})}x = \gamma_{0,t-h}\frac{1-\gamma_{t-h,t}^{2}}{1-\gamma_{0,t}^{2} + \delta^{2}\gamma_{0,t}^{2}}\bar{\mu} \\
        &+ \gamma_{t-h,t}\frac{1-\gamma_{0,t-h}^{2}+\delta^{2}\gamma_{0,t-h}^{2}}{1-\gamma_{0,t}^{2}+\delta^{2}\gamma_{0,t}^{2}}x = \mathbb{E}[X_{t-h}|X_{t}=x],
    \end{split}
    \end{equation}

    \begin{equation}
    \begin{split}
        &\Var{(\hat{X}_{t-h}|\hat{X}_{t}=x)} = (\sigma_{t,h}^{*})^{2}\ide = \left(\sigma_{t-h,t}^{2} + \nu_{t-h,t}^{2}\frac{\delta^{2}(1-\gamma_{0,t}^{2})}{1-\gamma_{0,t}^{2}+\delta^{2}\gamma^{2}_{0,t}}\right)\ide \\
        &=\left(\frac{(1-\gamma_{0,t-h}^{2})(1-\gamma_{t-h,t}^{2})}{1-\gamma_{0,t}^{2}}+\gamma_{0,t-h}^{2}\frac{\delta^{2}(1-\gamma_{t-h,t}^{2})^{2}}{(1-\gamma_{0,t}^{2})(1-\gamma_{0,t}^{2}+\delta^{2}\gamma^{2}_{0,t})}\right)\ide \\
        &=\frac{1-\gamma_{t-h,t}^{2}}{(1-\gamma_{0,t}^{2})(1-\gamma_{0,t}^{2}+\delta^{2}\gamma^{2}_{0,t})}\left((1-\gamma_{0,t-h}^{2})(1-\gamma_{0,t}^{2})+\delta^{2}\gamma_{0,t}^{2}(1-\gamma_{0,t-h}^{2})\right)\ide \\
        & + \frac{1-\gamma_{t-h,t}^{2}}{(1-\gamma_{0,t}^{2})(1-\gamma_{0,t}^{2}+\delta^{2}\gamma^{2}_{0,t})}\left(\delta^{2}\gamma_{0,t-h}^{2}(1-\gamma_{t-h,t}^{2})\right)\ide \\
        &=\frac{1-\gamma_{t-h,t}^{2}}{(1-\gamma_{0,t}^{2})(1-\gamma_{0,t}^{2}+\delta^{2}\gamma^{2}_{0,t})}\left((1-\gamma_{0,t-h}^{2})(1-\gamma_{0,t}^{2})+\delta^{2}\gamma_{0,t-h}^{2}(1-\gamma_{0,t}^{2})\right)\ide \\
        &=\frac{(1-\gamma_{t-h,t}^{2})(1-\gamma_{0,t-h}^{2}+\delta^{2}\gamma_{0,t-h}^{2})}{1-\gamma_{0,t}^{2}+\delta^{2}\gamma_{0,t}^{2}}\ide = \Var{(X_{t-h}|X_{t}=x)}.
    \end{split}
    \end{equation}

\end{proof}

\section{Reverse MR-VP SDE solver}
\label{app:mr_vp}
MR-VP DPM is characterized by the following forward and reverse diffusions:

\begin{equation}
\label{eq:mrvp_fwd_sde}
    dX_{t}=\frac{1}{2}\beta_{t}(\bar{X} - X_{t})dt + \sqrt{\beta_{t}}d\overrightarrow{W_{t}}\ ,
\end{equation}

\begin{equation}
\label{eq:mrvp_rev_sde}
    d\hat{X}_{t}=\left(\frac{1}{2}\beta_{t}(\bar{X} - \hat{X}_{t}) - \beta_{t}s_{\theta}(\hat{X}_{t}, \bar{X}, t)\right)dt + \sqrt{\beta_{t}}d\overleftarrow{W_{t}}.
\end{equation}

Using the same method as in Appendix \ref{app:sde_solution}, we can show that for $s<t$

\begin{equation}
    \Law{(X_{t}|X_{s})}=\mathcal{N}(\gamma_{s,t}X_{s}+(1-\gamma_{s,t})\bar{X},(1-\gamma_{s,t}^{2})\ide), \ \ \ \ \gamma_{s,t}=e^{-\frac{1}{2}\int_{s}^{t}{\beta_{u}du}}.
\end{equation}

With the following notation:

\begin{equation}
    \mu_{s,t}=\gamma_{s,t}\frac{1-\gamma_{0,s}^{2}}{1-\gamma_{0,t}^{2}}, \ \ \ \  \nu_{s,t}=\gamma_{0,s}\frac{1-\gamma_{s,t}^{2}}{1-\gamma_{0,t}^{2}}, \ \ \ \
    \sigma_{s,t}^{2}=\frac{(1-\gamma_{0,s}^{2})(1-\gamma_{s,t}^{2})}{1-\gamma_{0,t}^{2}}
\end{equation}

we can write down the parameters of Gaussian distribution $X_{s}|X_{t},X_{0}$:

\begin{equation}
    \mathbb{E}[X_{s}|X_{t},X_{0}]=\bar{X} + \mu_{s,t}(X_{t}-\bar{X}) + \nu_{s,t}(X_{0}-\bar{X}), \ \ \Var{(X_{s}|X_{t},X_{0})}=\sigma_{s,t}^{2}\ide.
\end{equation}

The Lemma \ref{lm:net_optimal} for MR-VP DPMs takes the following shape:

\begin{equation}
\begin{split}
    s_{\theta^{*}}(x,\bar{X},t) &= -\frac{1}{1-\gamma_{0,t}^{2}}\left(x - (1-\gamma_{0,t})\bar{X} - \gamma_{0,t}\mathbb{E}_{p_{0|t}(\cdot|x)}X_{0}\right) \\&
    =-\frac{1}{1-\gamma_{0,t}^{2}}\left((x - \bar{X}) - \gamma_{0,t}\left(\mathbb{E}_{p_{0|t}(\cdot|x)}X_{0} - \bar{X}\right)\right).
\end{split}
\end{equation}

The class of reverse SDE solvers we consider is 
\begin{equation}
    \hat{X}_{t-h} = \hat{X}_{t} + \beta_{t}h\left(\left(\frac{1}{2}+\hat{\omega}_{t,h}\right)(\hat{X}_{t} - \bar{X}) + (1+\hat{\kappa}_{t,h})s_{\theta}(\hat{X}_{t},\bar{X},t)\right) + \hat{\sigma}_{t,h}\xi_{t},
\end{equation}

where $t=1,1-h,..,h$ and $\xi_{t}$ are i.i.d. samples from $\mathcal{N}(0,\ide)$. Repeating the argument of the Theorem \ref{th:main} leads to the following optimal (in terms of likelihood of the forward diffusion sample paths) parameters:

\begin{equation}
\begin{split}
     \kappa_{t,h}^{*} =& \frac{\nu_{t-h,t}(1-\gamma_{0,t}^{2})}{\gamma_{0,t}\beta_{t}h} - 1,  \ \ \ \ 
     \omega_{t,h}^{*}=\frac{\mu_{t-h,t}-1}{\beta_{t}h} + \frac{1+\kappa_{t,h}^{*}}{1-\gamma_{0,t}^{2}} -\frac{1}{2},
     \\& (\sigma^{*}_{t,h})^{2} = \sigma^{2}_{t-h,t} + \frac{1}{n}\nu_{t-h,t}^{2}\mathbb{E}_{X_{t}}\left[\Tr{\left(\Var{(X_{0}|X_{t})}\right)}\right],
\end{split}
\end{equation}

which are actually the same as the optimal parameters (\ref{eq:notation_correction}) for VP DPM. It is of no surprise since MR-VP DPM and VP-DPM differ only by a constant shift.

\section{Reverse sub-VP SDE solver}
\label{app:sub_vp}

Sub-VP DPM is characterized by the following forward and reverse diffusions:

\begin{equation}
\label{eq:subvp_fwd_sde}
    dX_{t}=-\frac{1}{2}\beta_{t}X_{t}dt + \sqrt{\beta_{t}(1 - e^{-2\int_{0}^{t}{\beta_{u}du}})}d\overrightarrow{W_{t}}\ ,
\end{equation}

\begin{equation}
\label{eq:subvp_rev_sde}
    d\hat{X}_{t}=\left(-\frac{1}{2}\beta_{t}\hat{X}_{t} - \beta_{t}\left(1 - e^{-2\int_{0}^{t}{\beta_{u}du}}\right)s_{\theta}(\hat{X}_{t}, t)\right)dt + \sqrt{\beta_{t}(1 - e^{-2\int_{0}^{t}{\beta_{u}du}})}d\overleftarrow{W_{t}}.
\end{equation}

Using the same method as in Appendix \ref{app:sde_solution}, we can show that for $s<t$

\begin{equation}
    \Law{(X_{t}|X_{s})}=\mathcal{N}(\gamma_{s,t}X_{s},\left(1+\gamma_{0,t}^{4}-\gamma_{s,t}^{2}(1 + \gamma_{0,s}^{4})\right)\ide), \ \ \ \ \gamma_{s,t}=e^{-\frac{1}{2}\int_{s}^{t}{\beta_{u}du}}.
\end{equation}

Note that for $s=0$ this expression simplifies to

\begin{equation}
    \Law{(X_{t}|X_{0})}=\mathcal{N}(\gamma_{0,t}X_{0},(1 - \gamma_{0,t}^{2})^{2}\ide).
\end{equation}

With the following notation:

\begin{equation}
\begin{split}
    \mu_{s,t}=&\gamma_{s,t}\left(\frac{1-\gamma_{0,s}^{2}}{1-\gamma_{0,t}^{2}}\right)^{2}, \ \ \ \  \nu_{s,t}=\gamma_{0,s}\frac{1+\gamma_{0,t}^{4}-\gamma_{s,t}^{2}(1 + \gamma_{0,s}^{4})}{(1-\gamma_{0,t}^{2})^{2}}, \\&
    \sigma_{s,t}^{2}=\frac{(1-\gamma_{0,s}^{2})^{2}(1+\gamma_{0,t}^{4}-\gamma_{s,t}^{2}(1 + \gamma_{0,s}^{4}))}{(1-\gamma_{0,t}^{2})^{2}}
\end{split}
\end{equation}

we can write down the parameters of Gaussian distribution $X_{s}|X_{t},X_{0}$:

\begin{equation}
    \mathbb{E}[X_{s}|X_{t},X_{0}]=\mu_{s,t}X_{t} + \nu_{s,t}X_{0}, \ \ \Var{(X_{s}|X_{t},X_{0})}=\sigma_{s,t}^{2}\ide.
\end{equation}

The Lemma \ref{lm:net_optimal} for sub-VP DPMs takes the following shape:

\begin{equation}
    s_{\theta^{*}}(x,t) = -\frac{1}{(1-\gamma_{0,t}^{2})^{2}}\left(x - \gamma_{0,t}\mathbb{E}_{p_{0|t}(\cdot|x)}X_{0}\right).
\end{equation}

The class of reverse SDE solvers we consider is 
\begin{equation}
    \hat{X}_{t-h} = \hat{X}_{t} + \beta_{t}h\left(\left(\frac{1}{2}+\hat{\omega}_{t,h}\right)\hat{X}_{t} + (1+\hat{\kappa}_{t,h})\left(1-e^{-2\int_{0}^{t}{\beta_{u}du}}\right)s_{\theta}(\hat{X}_{t},t)\right) + \hat{\sigma}_{t,h}\xi_{t},
\end{equation}

where $t=1,1-h,..,h$ and $\xi_{t}$ are i.i.d. samples from $\mathcal{N}(0,\ide)$. Repeating the argument of the Theorem \ref{th:main} leads to the following optimal (in terms of likelihood of the forward diffusion sample paths) parameters:

\begin{equation}
\begin{split}
     \kappa_{t,h}^{*} = &\frac{\nu_{t-h,t}(1-\gamma_{0,t}^{2})}{\gamma_{0,t}\beta_{t}h(1+\gamma_{0,t}^{2})} - 1,  \ \ \ \ 
     \omega_{t,h}^{*}=\frac{\mu_{t-h,t}-1}{\beta_{t}h} + \frac{(1+\kappa_{t,h}^{*})(1+\gamma_{0,t}^{2})}{1-\gamma_{0,t}^{2}} -\frac{1}{2},
     \\& (\sigma^{*}_{t,h})^{2} = \sigma^{2}_{t-h,t} + \frac{1}{n}\nu_{t-h,t}^{2}\mathbb{E}_{X_{t}}\left[\Tr{\left(\Var{(X_{0}|X_{t})}\right)}\right].
\end{split}
\end{equation}

\section{Reverse VE SDE solver}
\label{app:ve}

VE DPM is characterized by the following forward and reverse diffusions:

\begin{equation}
\label{eq:ve_fwd_sde}
    dX_{t}=\sqrt{\left(\sigma_{t}^{2}\right)'}d\overrightarrow{W_{t}}\ ,
\end{equation}

\begin{equation}
\label{eq:ve_rev_sde}
    d\hat{X}_{t}=-\left(\sigma_{t}^{2}\right)'s_{\theta}(\hat{X}_{t}, t)dt + \sqrt{\left(\sigma_{t}^{2}\right)'}d\overleftarrow{W_{t}}.
\end{equation}

Since for $s<t$
\begin{equation}
    X_{t} = X_{s} + \int_{s}^{t}{\sqrt{\left(\sigma_{u}^{2}\right)'}d\overrightarrow{W_{u}}} \ ,
\end{equation}
similar argument as in Appendix \ref{app:sde_solution} allows showing that

\begin{equation}
    \Law{(X_{t}|X_{s})}=    \mathcal{N}\left(X_{s},\ide\cdot\int_{s}^{t}{\left(\sigma_{u}^{2}\right)'du}\right) = \mathcal{N}(X_{s},(\sigma_{t}^{2}-\sigma_{s}^{2})\ide).
\end{equation}

With the following notation:

\begin{equation}
    \mu_{s,t}=\frac{\sigma_{s}^{2}-\sigma_{0}^{2}}{\sigma_{t}^{2}-\sigma_{0}^{2}}, \ \ \ \  \nu_{s,t}=\frac{\sigma_{t}^{2}-\sigma_{s}^{2}}{\sigma_{t}^{2}-\sigma_{0}^{2}}, \ \ \ \
    \sigma_{s,t}^{2}=\frac{(\sigma_{t}^{2}-\sigma_{s}^{2})(\sigma_{s}^{2}-\sigma_{0}^{2})}{\sigma_{t}^{2}-\sigma_{0}^{2}},
\end{equation}

we can write down the parameters of Gaussian distribution $X_{s}|X_{t},X_{0}$:

\begin{equation}
    \mathbb{E}[X_{s}|X_{t},X_{0}]=\mu_{s,t}X_{t} + \nu_{s,t}X_{0}, \ \ \Var{(X_{s}|X_{t},X_{0})}=\sigma_{s,t}^{2}\ide.
\end{equation}

The Lemma \ref{lm:net_optimal} for VE DPMs takes the following shape:

\begin{equation}
    s_{\theta^{*}}(x,t) = -\frac{1}{\sigma_{t}^{2}-\sigma_{0}^{2}}\left(x - \mathbb{E}_{p_{0|t}(\cdot|x)}X_{0}\right).
\end{equation}

Repeating the argument of the Theorem \ref{th:main} leads to the following optimal (in terms of likelihood of the forward diffusion sample paths) reverse SDE solver:

\begin{equation}
    \hat{X}_{t-h} = \hat{X}_{t} + (\sigma_{t}^{2} - \sigma_{t-h}^{2})s_{\theta}(\hat{X}_{t},t) + \sigma_{t,h}^{*}\xi_{t},
\end{equation}

where

\begin{equation}
    (\sigma^{*}_{t,h})^{2} = \frac{(\sigma_{t}^{2}-\sigma_{t-h}^{2})(\sigma_{t-h}^{2}-\sigma_{0}^{2})}{\sigma_{t}^{2}-\sigma_{0}^{2}} + \frac{1}{n}\nu_{t-h,t}^{2}\mathbb{E}_{X_{t}}\left[\Tr{\left(\Var{(X_{0}|X_{t})}\right)}\right],
\end{equation}

$t=1,1-h,..,h$ and $\xi_{t}$ are i.i.d. samples from $\mathcal{N}(0,\ide)$.

\section{Toy examples}
\label{app:toy}
In this section we consider toy examples where data distribution $X_{0}$ is represented by a single point (corresponding to the case \textit{(ii)} of the Theorem \ref{th:main}) and by two points (corresponding to more general case \textit{(i)}). In the first case the point is unit vector $i=(1,1,..,1)$ of dimensionality $100$, in the second one two points $i$ and $-2i$ have the same probability. We compare performance of two solvers, Euler-Maruyama and the proposed Maximum Likelihood, depending on the number $N\in \{1,2,5,10,100,1000\}$ of solver steps. The output of the perfectly trained score matching network $s_{\theta^{*}}$ is computed analytically and Gaussian noise with variance $\varepsilon\in \{0.0,0.1,0.5\}$ is added to approximate the realistic case when the network $s_{\theta}$ we use is not trained till optimality. We considered VP diffusion model (\ref{eq:fwd_rev_sde}) with $\beta_{0}=0.05$ and $\beta_{1}=20.0$.

The results of the comparison are given in Table~\ref{tab:toy} and can be summarized in the following:
\begin{itemize}
    \item for both methods, larger $N$ means better quality;
    \item for both methods, more accurate score matching networks (smaller $\varepsilon$) means better quality;
    \item for large number of steps, both methods perform the same;
    \item it takes less number of steps for the proposed Maximum Likelihood solver to converge with a good accuracy to data distribution than it does for Euler-Maruyama solver;
    \item in accordance with the statement \textit{(ii)} of the Theorem \ref{th:main}, the optimal Maximum Likelihood solver leads to exact data reconstruction in the case when data distribution is constant and score matching network is trained till optimality (i.e. $\varepsilon=0.0$) irrespective of the number of steps $N$.
\end{itemize}

Also, in the second example where $X_{0}\in \{i,-2i\}$ the Maximum Likelihood SDE solver reconstructs the probabilities of these two points better than Euler-Maruyama which tends to output ``$i$-samples'' (which are closer to the origin) more frequently than ``$-2i$-samples''. E.g. for $\varepsilon=0.0$ and $N=10$ the frequency of ``$i$-samples'' generated by Euler-Maruyama scheme is $54\%$ while this frequency for Maximum Likelihood scheme is $50\%$ ($500k$ independent runs were used to calculate these frequencies).

\begin{table}
\caption{Maximum Likelihood (ML) and Euler-Maruyama (EM) solvers comparison in terms of Mean Square Error (MSE). MSE $<0.001$ is denoted by \textit{conv}, MSE $>1.0$ is denoted by \textit{div}. $N$ is the number of SDE solver steps, $\varepsilon$ is variance of Gaussian noise added to perfect scores $s_{\theta^{*}}$.}
\begin{center}
\begin{tabular}{|c|c|c|c|c|c|c|}
\hline
MSE &
\multicolumn{3}{c|}{\begin{tabular}[c]{@{}c@{}c@{}}$X_{0}=\{i\}$\end{tabular}} &
\multicolumn{3}{c|}{\begin{tabular}[c]{@{}c@{}c@{}}$X_{0}=\{i, -2i\}$\end{tabular}} \\ \cline{2-7}
ML / EM
&$\varepsilon=0.0$ &$\varepsilon=0.1$ &$\varepsilon=0.5$ &$\varepsilon=0.0$ &$\varepsilon=0.1$ &$\varepsilon=0.5$ \\ \hline
$N=1$ &\textit{conv} / \textit{div} &\textit{div} / \textit{div} &\textit{div} / \textit{div} &\textit{div} / \textit{div} &\textit{div} / \textit{div} &\textit{div} / \textit{div}\\ \hline
$N=2$ &\textit{conv} / \textit{div} &\textit{div} / \textit{div} &\textit{div} / \textit{div} &$0.15$ / \textit{div} &\textit{div} / \textit{div} &\textit{div} / \textit{div}\\ \hline
$N=5$ &\textit{conv} / \textit{div} &$0.017$ / \textit{div} &$0.085$ / \textit{div} &\textit{conv} / \textit{div} &$0.017$ / \textit{div} &$0.085$ / \textit{div}\\ \hline
$N=10$ &\textit{conv} / $0.57$ &$0.001$/$0.59$ &$0.005$/$0.67$ &\textit{conv} / $0.57$ &$0.001$/$0.59$ &$0.006$/$0.67$\\ \hline
$N=100$ &\textit{conv} / $0.01$ &\textit{conv} / $0.01$ &\textit{conv} / $0.01$ &\textit{conv} / $0.01$ &\textit{conv} / $0.01$ &\textit{conv} / $0.01$\\ \hline
$N=1000$ &\textit{conv} / \textit{conv} &\textit{conv} / \textit{conv} &\textit{conv} / \textit{conv} &\textit{conv} / \textit{conv} &\textit{conv} / \textit{conv} &\textit{conv} / \textit{conv}\\ \hline
\end{tabular}
\end{center}
\label{tab:toy}
\end{table}

\section{Speaker conditioning network}
\label{app:h}
The function $x\cdot tanh(softplus(x))$ is used as a non-linearity in the speaker conditioning network $g_{t}(Y)$. First, time embedding $t_{e}$ is obtained by the following procedure: time $t\in [0,1]$ is encoded with positional encoding \citep{Song-main}, then resulting $256$-dimensional vector $t'$ is passed through the first linear module with $1024$ units, then a non-linearity is applied to it and then it is passed through the second linear module with $256$ units. Next, noisy mel-spectrogram $Y_{t}$ for \textit{wodyn} input type or $Y_{t}$ concatenated with $\{Y_{s}|s=0.5/15,1.5/15,..,14.5/15\}$ for \textit{whole} is passed through $6$ blocks consisting of $2$D convolutional layers each followed by instance normalization and Gated Linear Unit. The number of input and output channels of these convolutions is $(1, 64), (32, 64), (32, 128), (64, 128), (64, 256), (128, 256)$ for \textit{wodyn} input type and the same but with $16$ input channels in the first convolution for \textit{whole} input type. After the $2$nd and $4$th blocks $MLP_{1}(t_{e})$ and $MLP_{2}(t_{e})$ are broadcast-added where $MLP_{1}$ ($MLP_{2}$) are composed of a non-linearity followed by a linear module with $32$ ($64$) units. After the last $6$th block the result is passed through the final convolution with $128$ output channels and average pooling along both time and frequency axes is applied resulting in $128$-dimensional vector. All convolutions except for the final one have (kernel, stride, zero padding) = $(3,1,1)$ while for the final one the corresponding parameters are $(1,0,0)$. Denote the result of such processing of $Y$ by $c$ for \textit{wodyn} and \textit{whole} input types.

Clean target mel-spectrogram $Y_{0}$ is used to obtain $256$-dimensional speaker embedding $d$ with the pre-trained speaker verification network \citep{dvector} which is \textit{not} trained. Vectors $d$, $c$ and $t'$ are concatenated (except for \textit{d-only} input type where we concatenate only $d$ and $t'$), passed through a linear module with $512$ units followed by a non-linearity and a linear module with $128$ units. The resulting $128$-dimensional vector is the output of the speaker conditioning network $g_{t}(Y)$.

\section{Training hyperparameters and other details}
\label{app:details}
Encoders and decoders were trained with batch sizes $128$ and $32$ and Adam optimizer with initial learning rates $0.0005$ and $0.0001$ correspondingly. Encoders and decoders in VCTK models were trained for $500$ and $200$ epochs respectively; as for LibriTTS models, they were trained for $300$ and $110$ epochs. The datasets were downsampled to $22.05$kHz which was the operating rate of our VC models. VCTK recordings were preprocessed by removing silence in the beginning and in the end of utterances. To fit GPU memory, decoders were trained on random speech segments of approximately $1.5$ seconds rather than on the whole utterances. Training segments for reconstruction and the ones used as input to the speaker conditioning network $g_{t}(Y)$ were different random segments extracted from the same training utterances. Noise schedule parameters $\beta_{0}$ and $\beta_{1}$ were set to $0.05$ and $20.0$. 

Our VC models operated on mel-spectrograms with $80$ mel features and sampling rate $22.05$kHz. Short-Time Fourier Transform was used to calculate spectra with $1024$ frequency bins. Hann window of length $1024$ was applied with hop size $256$.

For \textit{Diff-LibriTTS} models we used simple spectral subtraction algorithm in mel domain with spectral floor parameter $\beta=0.02$ as post-processing to reduce background noise sometimes produced by these models. Noise spectrum was estimated on speech fragments automatically detected as the ones corresponding to silence in source mel-spectrogram.

\section{Details of AMT tests}
\label{app:subj}
For fair comparison with the baselines all the recordings were downsampled to $16$kHz; we also normalized their loudness. In speech naturalness tests workers chosen by geographic criterion were asked to assess the overall quality of the synthesized speech, i.e to estimate how clean and natural (human-sounding) it was. Five-point Likert scale was used: $1$ - ``Bad'', $2$ - ``Poor'', $3$ - ``Fair'', $4$ - ``Good'', $5$ - ``Excellent''. Assessors were asked to wear headphones and work in a quiet environment. As for speaker similarity tests, workers were asked to assess how similar synthesized samples sounded to target speech samples in terms of speaker similarity. Assessors were asked not to pay attention to the overall quality of the synthesized speech (e.g. background noise or incorrect pronunciation). 
Five-point scale was used: $1$ - ``Different: absolutely sure'', $2$ - ``Different: moderately sure'', $3$ - ``Cannot decide more same or more different'', $4$ - ``Same: moderately sure'', $5$ - ``Same: absolutely sure''.

\end{document}